\documentclass[10pt]{article}
\usepackage[utf8]{inputenc}
\usepackage{subfiles}
\usepackage[english]{babel}
\usepackage[a4paper,top=3 cm,bottom=3 cm,inner=2.5 cm,outer=2.5 cm]{geometry} 
\usepackage{graphicx}
\usepackage{bm}
\usepackage{mathtools}
\usepackage{amsmath}
\usepackage{amsfonts}
\usepackage{subcaption}
\usepackage{mathbbol}
\usepackage{enumitem}
\usepackage{cancel}
\DeclareMathOperator{\Tr}{Tr}
\setlist[itemize]{label=\textbullet}
\usepackage{float} 
\usepackage[table]{xcolor} 
\usepackage{booktabs} 
\usepackage[utf8]{inputenc} 
\usepackage{hyperref}
\usepackage{amsmath, amssymb, amsthm} 
\usepackage{eurosym}
\usepackage{tikz}
\newtheorem{thm}{Theorem}[section] 
\newtheorem{lem}[thm]{Lemma}
\newtheorem{prop}[thm]{Proposition} 

\newtheorem{defn}{Definition} [section]
\newtheorem{rem}{Remark}[section]

\theoremstyle{definition}

\usepackage{wasysym}
\hypersetup{pdfborder={0 0 0}} 
\usepackage{url} 
 
\usepackage{titlesec}
\titleformat{\paragraph}
{\normalfont\normalsize\bfseries}{\theparagraph}{1em}{}
\titlespacing*{\paragraph}
{0pt}{3.25ex plus 1ex minus .2ex}{1.5ex plus .2ex}

\titleformat{\subparagraph}
{\normalfont\normalsize\bfseries}{\theparagraph}{1em}{}
\titlespacing*{\paragraph}
{0pt}{3.25ex plus 1ex minus .2ex}{1.5ex plus .2ex}

\usepackage{mathrsfs}
\usepackage{leftidx}
\usepackage{braket}
\usepackage{physics}

\newcommand{\di}{\text{d}}
\newcommand{\en}[1]{\omega_{\textbf{#1}}}

\newcommand\numberthis{\addtocounter{equation}{1}\tag{\theequation}}

\usepackage{csquotes}
\usepackage[
backend=biber,
giveninits=true,
style=alphabetic,
sorting = nty,
isbn = false,
url = false,
maxbibnames = 4,
]{biblatex}
\begin{document}
\par    
\bigskip  
\noindent 
\LARGE{\bf Perturbative Construction of Equilibrium States for Interacting Fermionic Field Theories}

\vspace*{0.3cm}
\noindent
\large{\bf Semiclassical Maxwell Equation and the Debye Screening Length}
\bigskip \bigskip
\par 
\rm 
\normalsize 
 
\large
\noindent 
{\bf Stefano Galanda$^{1,a}$}\\
\par
\small

\noindent$^1$ Dipartimento di Matematica, Universit\`a di Genova - Via Dodecaneso, 35, I-16146 Genova, Italy. \smallskip
\smallskip

\noindent E-mail: 
$^a$stefano.galanda@dima.unige.it

\normalsize
${}$ \\ \\
 {\bf Abstract} \ \
In this paper, we aim to extend to interacting massive and massless fermionic theories the recent perturbative construction of equilibrium states developed within the framework of perturbative algebraic quantum field theory on Lorentzian spacetime. 
We analyze the case of interactions which depend on time by a smooth switch-on function and on space by a suitably bounded function that multiplies an interaction Lagrangian density constructed with the field of the theory.
The construction is achieved by first considering the case of compact support and, in a second step, by removing the space cutoff 
with a suitable limit (adiabatic limit). 
As an application, we consider a Dirac field interacting with a classical stationary background electromagnetic potential, and we compute at first perturbative order (linear response) the expectation value of the conserved current on the equilibrium state for the interacting theory. 
The resulting expectation value is written as a convolution, in the space coordinates, between the electromagnetic potential and an integral kernel which, at vanishing conjugate momentum, gives the inverse of the square Debye screening length at finite temperature. 
The corresponding Debye screening effect is visible
in the backreaction treated semiclassically of this current on the classical background electromagnetic potential sourced by a classical external current.
\bigskip
${}$

\section{Introduction}
To understand the thermodynamical properties of a system with respect to a certain dynamics, the starting point is the description of its equilibrium configurations. More specifically, an equilibrium configuration is a stationary configuration to which it is possible to associate certain macroscopic quantities like the temperature. In describing systems with finitely many degrees of freedom, on a finite dimensional Hilbert space $\mathcal{H}$, an equilibrium configuration with respect to a given time evolution (dynamics) generated by a Hamiltonian $H$ is represented by a Gibbs state at inverse temperature $\beta$. These, are positive of unit trace and hermitian matrices of the kind $\rho_{\beta} = \mathcal{N} e^{-\beta H}$, where $\mathcal{N}$ is a normalization factor.
They implicitly define a positive normalized and linear functional (state) that acts on any $A \in \mathcal{B}(\mathcal{H})$ as $\left \langle A \right \rangle_{\beta} \coloneqq \Tr(\rho_{\beta} A)$. A careful analysis of the analyticity properties of Gibbs states led to a characterization of equilibrium configurations, first done by Kubo, Martin, and Schwinger (KMS) \cite{Kubo, MartinSchwinger}, via properties of the associated functionals with respect to the considered dynamics (see Definition \ref{def: KMScond} for further details). However, as a property of states, the KMS condition defines equilibrium in the most general setting. Indeed, given a $*$-dynamical system $(\mathfrak{A},\tau_t)$ consisting of a $*$-algebra $\mathfrak{A}$ and a one parameter family of $*$-automorphisms $\tau_t$ on $\mathfrak{A}$, a state $\omega: \mathfrak{A} \to \mathbb{C}$ is defined to be of equilibrium at inverse temperature $\beta$ with respect to $\tau_t$ if it satisfies the KMS condition with parameter $\beta$ also known as $\beta$-KMS condition.\\
The stability of the equilibrium configurations under perturbations of the dynamics, naturally arises when one aims at studying open systems, equilibrium systems in thermal contact, more generally systems involving interactions among the constituents. Indeed, denoting by $\tau_t^V$ the perturbed dynamics, a $\beta$-KMS state $\omega^{\beta}$ with respect to $\tau_t$ is in general not of equilibrium with respect to $\tau_t^V$. In quantum statistical mechanics where $\mathfrak{A}$ is a $C^*$-algebra, it was shown that if a certain clustering condition holds, see e.g. \cite{BratelliKishimotoRobinson, BratteliRobinson}, in the limit $t \to \infty$ the expectation value $\omega^{\beta}(\tau_t^V(A))$ converges to a specific $\beta$-KMS state for $\tau_t^V$ denoted $\omega^{\beta,V}(A)$. The latter, first defined by Araki in \cite{Araki}, is a state that satisfies the KMS condition for $\tau_t^V$.\\
The corresponding existence of equilibrium states for interacting quantum field theories, namely when $\mathfrak{A}$ is only a $*$-algebra, suffers in the traditional approaches of spurious infrared divergences at the higher loop, see e.g. \cite{Steinmann}. Nevertheless, for an interacting massive scalar field theory on Minkowski spacetime, these issues were circumvented by Fredenhagen and Lindner \cite{FredenhagenLindnerKMS_2014, Lindner} generalizing the mentioned work of Araki in the setting of $*$-dynamical system. In this setup, in order to define the interacting dynamics, it is crucial the existence of the Bogoliubov map. This defines an embedding of the algebra of the interacting theory as formal power series into the algebra of smeared Wick polynomials of the free theory (see Equation \eqref{eq: Bogol}). It is thanks to this map that $\tau_t^V$ is defined via pullback of the free dynamics $\tau_t$ with respect to the corresponding Bogoliubov map (see Equation \eqref{def: InterDynam}).\\
In the works of Fredenhagen and Lindner it is stressed, with the caveat that the equilibrium state for the interacting theory and the dynamics are defined just as formal power series, how the existence of equilibrium configurations needs to be carefully studied when the interactions have arbitrary spatial support. Furthermore, in \cite{FaldinoEquilibriumpAQFT}, it was shown that also the stability of the constructed interacting KMS states is crucially related to the spatial support of the interaction. Indeed, the corresponding limit $t \to \infty$ of $\omega^{\beta}(\tau_t^V(A))$ is shown to converge to the $\beta$-KMS state with respect to $\tau_t^V$ when the support of the perturbation $V$ is spatially compact otherwise, in general, counterexamples can be constructed.\\
Extending the result beyond the massive scalar field arises as a natural quest. Equilibrium states at positive temperature were constructed for massless scalar fields in \cite{DragoHackPin} using an idea similar to that of the thermal mass for interacting $\phi^4$-theories. The latter is the idea to use the free KMS two-point function in the definition of Wick polynomials to gain a mass term in the interacting Lagrangian that, by a partial series resummation, is included in the free theory.\\

In this paper, we extend the results of \cite{FredenhagenLindnerKMS_2014, Lindner} to the case of $\mathfrak{A}$ the $*$-algebra describing the Wick polynomials of a free fermionic field theory on Minkowski spacetime, including the massless case treated without the need of the thermal mass argument. This is achieved, as fermionic theories possess a better infrared behavior both at positive and zero temperature as a consequence of two key features: the Fermi as opposed to the Bose factor and the additional derivatives in the two-point functions \eqref{eq: 2puntiKMS},\eqref{eq: 2puntiGround}.\\
Nonetheless, the construction of equilibrium states for interacting theories, especially fermionic at zero temperature, was non-perturbatively studied in various models in different spacetime dimensions relying on Renormalization Group techniques reviewed in \cite{BenfattoGallavotti, Mastropietro, Salmhofer}. In $d=2$ spacetime dimensions we mention the infrared Gross-Neveu model \cite{GawedzkiKupiainen, FMRS} and the Thirring model \cite{BenfattoFalcoMastropietro}. While, in $d=4$ spacetime dimension, the infrared massive QED \cite{MastropietroQED4}. However, compared to these works, the perturbative results of \cite{FredenhagenLindnerKMS_2014} produce well-defined non-trivial interacting equilibrium states for the massive scalar field, also at positive temperature, both in the infrared and ultraviolet limits. The latter, compared to the aforementioned results, introduces the cutoffs in position space. A space cutoff $h \in \mathcal{C}_0^{\infty}(\mathbb{R}^3)$ and a time cutoff $\chi \in \mathcal{C}^{\infty}(\mathbb{R})$ are used in the definition of the interaction. After the construction of the equilibrium states, the dependence on the cutoffs is discussed. The space cutoff is proven, for the mentioned theories, to be removable at the level of expectation values showing the existence of the limit $h \to 1$ at each perturbative order. The time cutoff is shown to be arbitrary in the family of switching functions $\mathcal{J}_{\epsilon} \coloneqq \left\{ \chi \in \mathcal{C}^{\infty}(\mathbb{R}) : \chi = 0 \,\, \mathrm{on} \,\, (-\infty, -2\epsilon) \, , \, \chi = 1 \,\, \mathrm{on} \,\, [-\epsilon, +\infty) \right\}$ for $\epsilon > 0$, in the sense that a different choice $\chi'$, even for different $\epsilon' \neq \epsilon$, leads to the exact same equilibrium state.\\
Despite our formalism being perturbative, we lack a systematic non-perturbative quantization procedure for general theories, we can treat almost any kind of interacting theory in four spacetime dimensions. Indeed, the interacting Lagrangian $\mathcal{L}_I$ entering the total action of the theory needs just to be invariant under spacetime translations up to possibly a sufficiently regular function of the space coordinates playing the role of an external potential. For example the case $\mathcal{L}_I = (\overline{\psi}\gamma^{\mu} \psi)^n$ or $\mathcal{L}_I = \overline{\psi}\cancel{P}(\mathbf{x}) \psi$, where $P(\mathbf{x})$ is smooth function describing a possible external potential (see Theorem \ref{thm: 1.1} for the general case) Moreover, the dependence on the spacetime background can directly be investigated. Nevertheless, the noteworthy existence of non-perturbative results in various models should be compared with the perturbative calculations capturing in this way any kind of non-perturbative feature.\\

In tackling the construction of equilibrium states for interacting fermionic theories, we need to take into account that, in dimension greater than $2$, the canonical quantization formalism doesn't allow to describe a theory with an interaction that is switched on in time \cite{Powers}. Therefore, to avoid the restriction to a Cauchy surface, we will make use of the time-slice property. This asserts that any observable on the whole Minkowski spacetime can be written as a sum of an observable vanishing on-shell and another supported in a small neighborhood (the time-slice) of an arbitrary Cauchy surface. Therefore, the construction can be performed restricting on the slice instead of on the sharper Cauchy surface. However, due to the fermionic statistics, in order to prove the validity of the time-slice property we lack the causal commutation between arbitrary observables. For this reason, following \cite{DutschBook, BrunettiDuestschRejznerFredenhagenFermions}, we turn the statistic into that of bosons using a deformation of the algebraic structure. Such a deformation, referred to in literature also as $\eta$\textit{-trick} \cite{ItzyksonZuber}, is done by an appropriate "juxtaposition" with Grassmann numbers. Nevertheless, such a juxtaposition needs to be physically irrelevant. Therefore, as it is proper of the algebraic formalism, we discuss in this paper a simple but innovative way to remove such additional degrees of freedom by a prescription on states over the twisted algebraic structure.\\

The main result is summarized here in compact form (see Theorem \ref{thm: 1} and  \ref{thm: 2} for the complete statement).
\begin{thm}\label{thm: 1.1}
Let $\omega^{\beta}$ be the equilibrium state at inverse temperature $0 < \beta \leq \infty$ with respect to the free dynamics $\tau_t$ on the $*$-algebra $\mathfrak{A}$ of smeared Wick polynomials of the free Dirac field on Minkowski spacetime. Consider a perturbation of the action of the free theory given by a potential $V \in \mathfrak{A}$ of the form:
\begin{equation*}
    V = \int_{\mathbb{M}} \di^4x  \chi(t) h(\mathbf{x}) \mathcal{L}_I.
\end{equation*}
Here $h \in \mathcal{C}_0^{\infty}(\mathbb{R}^3)$ is a space cutoff, $\chi$ a switch-on function as in \eqref{eq:accad} and $\mathcal{L}_I$ is the interacting Lagrangian that can either be of the kind $\mathcal{L}_{I} = (\overline{\psi} \psi)^n$ or $\mathcal{L}_I = (\overline{\psi} \cancel{P}(\mathbf{x}) \psi)^{n}$ for $n \in \mathbb{N}$ and $P: \mathbb{R}^3 \to \mathbb{R}$ so that $V \subset \mathscr{F}_{\mu c}$ (see \eqref{eq: microcausal}). Then, there exists a state on $\mathfrak{A}$ of equilibrium $\omega^{\beta,V}$ with respect to the perturbed dynamics $\tau_t^V$ (see \eqref{def: InterDynam}) at each perturbative order. Finally, $\omega^{\beta,V}$ has a convergent adiabatic limit ($h \to 1$) at each fixed perturbative order.
\end{thm}

As an application of the presented construction of equilibrium states we perform a computation in \textit{linear response theory}. The latter, in the form of the Kubo formula, was proven to be valid for interacting fermions in the non-perturbative context of quantum many body theory in the adiabatic limit \cite{BachmannDeRoeckFraas, MonacoTeufel} also at positive temperature \cite{GreenblattLangeMarcelliPorta}. Here, compared to the aforementioned works, despite the computation being perturbative, we study the linear response of a system in equilibrium at positive temperature already in the continuum. Therefore, in our formalism, the independence of the linear response from the switching process is already implicit in the construction of the interacting state $\omega^{\beta,V}$.\\
In particular, we are interested in computing the Debye length relevant to the study of plasmas. A plasma, in QFT, is described by a charged field at equilibrium interacting with an external electromagnetic field\footnote{Similar results are obtained also for non-Abelian gauge theories like QCD.}. The first interesting question that can be addressed is how the electromagnetic field is affected by the presence of a Dirac field at thermal equilibrium. For a stationary external field, the result is an exponential screening of the field with characteristic decay length $\lambda_D$ called \textit{Debye screening length} \cite{DebyeHuckel}, effectively interpreted as if the photon propagator acquires mass (\textit{Debye mass}). In particular, as the Debye length is experimentally measured in the case of various ion gases, the literature on it is very rich. In Thermal Field Theory, $\lambda_D$ is obtained by computing the corrections to the photon propagator at high temperatures $T$ and considering the ultrarelativistic limit (negligible mass for the Dirac field) \cite{LeBellac, KaputsaGale, Altherr}. Similar analyses are performed in the Constructive Field Theory approach \cite{BrydgesFederbush, Brydges, BrydgesKeller, FrohlichPark} and in \cite{BlaizotIancuParwani, ArnoldYaffe, Schneider, Rebhan}, see also further citations therein.\\

In this paper, our analysis starts from the computation of the expectation value of the conserved current associated with the Dirac field on the state $\omega^{\beta,V}$ of equilibrium concerning the QED coupling. The electromagnetic potential $A^{\mu}$ is assumed to be a componentwise bounded function possibly not smooth but with singular behavior determined via a specific condition on its wavefront set (see discussion after Equation \eqref{eq: couplQED}). To compute this expectation value in linear response (first perturbative order), we adopt a more general approach compared to the aforementioned works of Thermal Field Theory. Indeed, we perturbatively expand both the thermal equilibrium states and the observables. As shown in \cite{JoaoNicNic}, this approach reduces to the real-time and imaginary-time formalism in specific limits. The reason for this choice is the impossibility of using real-time or imaginary-time formalism alone to compute the screening of an electromagnetic field supported everywhere due to the arising of secular growths in the perturbations \cite{GalandaSangalettiPin}. As discussed in \cite{LandsmanVanWert} and proved in \cite{FaldinoEquilibriumpAQFT}, the integration over Schwinger-Keldysh contour can be simplified in either a real or an imaginary integral only in specific cases. Within our approach, the spurious infrared divergences arising in the real-time formalism, manifesting as secular effects, are absent. Here, as we aim at considering a possibly infinitely extended stationary electromagnetic field interacting with a Dirac field at arbitrary temperature $T$ and mass $m$, we have to adopt this more general approach.\\

The result that we get is summarized in the following proposition:

\begin{prop}
Let $\omega^{\beta,V}$ be the equilibrium state constructed above with interaction
\begin{equation*}
    V = \int_{\mathbb{M}} \di^4 x \chi(t) h(\mathbf{x}) A_{\mu}(\mathbf{x})\overline{\psi}(x) \gamma^{\mu} \psi(x).
\end{equation*}
Then, denoting by $R_V(j^{\mu})$ the interacting observable associated with the conserved current of the free Dirac field, at first perturbative order $\omega^{\beta,V}(R_V(j^{\mu}))$ is a convolution of $A^{\mu}$ with something that at zeroth order in the conjugate momentum reduces to a Dirac delta distribution:
\begin{equation*}
    \omega^{\beta,V}(R_V(j^{\mu})(\mathbf{x})) = \left(A^{\mu} \ast m_D^2 \delta\right)(\mathbf{x})
\end{equation*}
where:
\begin{equation}\label{eq: DebyeCostanti}
    m_D^2 = \frac{4e^2 m^2 \hbar c}{(2 \pi)^2 \epsilon_0} \sum_{n=0}^{\infty} (-1)^{n} K_2\big( (n+1) \beta m c^2\big),
\end{equation}
is the square Debye mass with $K_2(z)$ the modified Bessel functions of second kind and index $2$.
\end{prop}

The above square Debye mass is a result valid at any inverse temperature $\beta$ of the equilibrium configuration and mass $m$ of the Dirac field. Indeed, in the ultrarelativistic limit (massless Dirac field) and (or) at high temperature, the results mentioned in the above literature are recovered by our result \eqref{eq: DebyeCostanti}. To the best of our knowledge, this is the first generalization of the results of Thermal Field Theory in this direction.\\
Furthermore, in support of the interpretation of \eqref{eq: DebyeCostanti} as the Debye mass, we compute the influence on the electromagnetic potential of the coupling with the Dirac field. However, instead of evaluating loop corrections to the photon propagator, we directly solve the \textit{semiclassical} Maxwell equations:
\begin{equation}\label{eq: SemiclIntro}
    \Delta A^{\mu}(\mathbf{x}) = -\left\langle j^{\mu}_{\mathrm{q}}(\mathbf{x})\right\rangle_{\beta,V} - j^{\mu}_{\mathrm{class}}(\mathbf{x}),
\end{equation}
where $j^{\mu}_{\mathrm{class}} \in \mathcal{C}^{\infty}_0(\mathbb{M}, T\mathbb{M})$ is a stationary classical source and $\left\langle j^{\mu}_{\mathrm{q}}\right\rangle_{\beta,V}$ a quantized source given by the above computed expectation value. The solution, valid for any $\beta$ and $m$, is a convolution between the classical source and a correction that at zeroth order in the conjugate momentum is a Yukawa type potential with characteristic length $\lambda_D$ given by the inverse of $m_D$. Therefore, the classical solution undergoes a backreaction that at zeroth order in the external momentum consists in a screening with characteristic length given by the square root of the inverse of \eqref{eq: DebyeCostanti}.\\

The organization of the paper is as follows. In the next Section, we shortly review the perturbative approach to interacting theories in AQFT and discuss how the algebras are deformed turning the fermionic statistics of its elements into a bosonic one. In Section \ref{sec: SuppComp} the construction of equilibrium states for spatially compact interactions is performed, while in Section \ref{sec: LimAd} their adiabatic limit is studied. Finally, Section \ref{sec: DebyeLeng} is devoted to the application of the proven results to solve the semiclassical Maxwell equations computing the screening of the electromagnetic field due to its propagation in the thermal Dirac field bath. Finally, the technical proofs together with further necessary technical results are collected in the Appendix.

${}$ \\ \\ \\
{\bf  Acknowledgments}
I am grateful to the referee for many valuable hints that improved this version of the paper. Moreover, I thank the National Group of Mathematical Physics (GNFM-INdAM) for the support. My research is supported by a PhD scolarship of the Math Department of University of Genoa awarded with the MIUR Excellence Department Project 2023-2027, CUP\textunderscore$\,$D33C23001110001.\\
I am grateful to Nicola Pinamonti for his support and guidance during the writing of this paper. I also thank Tommaso Bruno (to whom I am grateful for the help in proving Lemma \ref{lem: 2}), Andrea Bruno Carbonaro, Edoardo D'Angelo, Nicolò Drago, Paolo Meda, Simone Murro and Gabriel Schmid for the useful discussions.

\section{Perturbative Algebraic Quantum Field Theory for fermionic Fields}
We consider theories propagating on Minkowski spacetime $(\mathbb{M} = \mathbb{R} \times \Sigma, \mathbf{g})$ with metric signature $(-,+,+,+)$. Moreover, we are concerned with the quantization of field theories whose Lagrangian density is factorized:
\begin{equation*}
    \mathcal{L} = \mathcal{L}_F + \mathcal{L}_I
\end{equation*}
where the factors are respectively the free ($\mathcal{L}_F = \overline{\psi}(i\cancel{\partial} - m)\psi$) and interacting Lagrangian density describing massive ($m > 0$) or massless ($m = 0$) interacting Dirac fields.\\
In what follows we present how interacting quantum field theories are perturbatively treated in the algebraic setting. The interested reader is referred to the original works \cite{BrunettiFredenhagen00, BFV03, HollandsWald2001, HW02, HW05, FredenhagenRejzner2012} and to the recent detailed reviews contained in the books \cite{DutschBook, KasiaBook}. This approach combines the techniques of renormalization of perturbative interacting theories \cite{EpsteinGlaser} with the axiomatic algebraic approach to quantum field theories \cite{HaagKastler, HaagLQP, BuchFredAQFT}. This is an approach to quantized field theories on Lorentzian background spacetime that puts the focus on the observables, represented as elements of algebras, assigned to open regions of spacetime. Locality is encoded in (anti-)commutation relations between the observables of algebras associated to different regions of spacetime. Once the algebras of the free theory have been constructed, those of the interacting theory are mapped into the $*$-algebra of the free theory as formal power series in the coupling constant. States are constructed in a second step assigning positive, normalized and linear functional over the algebras. In addition, fixing a state, the \textit{GNS theorem} allows to represent the algebra of observables as operators on a Hilbert space. One of the advantages of this formulation, is that renormalization turns out to be automatically independent from the choice of a specific state.

\subsection{Algebra of Fermionic Functionals}
Following \cite{DutschBook}, we introduce the algebras describing anti-commuting Fermi fields. Let us consider the vector space $\Gamma(\mathbb{M}, \mathbb{C}^4)$ of smooth sections of the vector bundle constructed over Minkowski spacetime $\mathbb{M}$ with fibers $\mathbb{C}^4$. On it, we assume to have a fiberwise irreducible complex representation $\pi$ of the Clifford algebra $Cl_{1,3}$ as elements in $GL(4,\mathbb{C})$ satisfying:
\begin{equation*}
    [\pi(l^{\mu}), \pi(l^{\nu})]_+ = \pi(l^{\mu}) \pi(l^{\nu}) + \pi(l^{\nu}) \pi(l^{\mu}) = -2 \mathbf{g}^{\mu \nu} \mathbb{1}_{4 \times 4},
\end{equation*}
for $\{l^{\mu}\}_{\mu = 0,1,2,3}$ the generators of the Clifford algebra. In particular, the following representation is chosen:
\begin{equation*}
    \pi(l^0) = \gamma^{0} \coloneqq \begin{pmatrix}
    0 & -\mathbb{1}_{2\times 2}\\
    -\mathbb{1}_{2\times 2} & 0
    \end{pmatrix} \quad \pi(l^i) = \gamma^{i} \coloneqq \begin{pmatrix}
    0 & -\sigma^i\\
    \sigma^i & 0
    \end{pmatrix}
\end{equation*}
where the $\{ \gamma^{\mu} \}_{\mu = 0,1,2,3}$ are referred to as Dirac matrices and $\sigma^i \in GL(2,\mathbb{C})$ are the Pauli matrices. Furthermore, each fiber, is equipped with a non-degenerate Lorentz invariant sesquilinear form $\mathbb{C}^4 \times \mathbb{C}^4 \ni (v_1(x), v_2(x)) \mapsto \overline{v_1} v_2(x) \in \mathbb{C}$ where:
\begin{equation*}
    \overline{v}(x) \coloneqq v^{\dagger}(x) \gamma^0,
\end{equation*}
with $v^{\dagger}(x) \in (\mathbb{C}^4)^*$ the adjoint of $v(x) \in \mathbb{C}^4$ with respect to the standard product of $\mathbb{C}^4$. The fiberwise definition of $\overline{v}(x)$ extends to $\overline{v} \in \Gamma(\mathbb{M}, (\mathbb{C}^4)^*)$. In particular, elements of $\Gamma(\mathbb{M}, \mathbb{C}^4)$ and $\Gamma(\mathbb{M}, (\mathbb{C}^4)^*)$ are considered to be independent.\\
Finally, with the aid of the Dirac matrices, given $v \in \Gamma(\mathbb{M}, \mathbb{C}^4)$, we introduce the fiberwise defined notation $\cancel{v} = \gamma^{\mu} v_{\mu}$.\\
$X = \Gamma(\mathbb{M}, \mathbb{C}^4) \oplus \Gamma(\mathbb{M}, (\mathbb{C}^4)^*)$ is referred to as \textit{configuration space} and the choice is motivated by the considered theory. Indeed, the following definitions, generalize to sections of general vector bundles.
\begin{defn}
We call \textbf{fermionic functionals} on the configuration space $X$, any multi-linear functional on the exterior algebra $\bigwedge X \coloneqq \bigoplus_{n=0}^{\infty} \bigwedge^n X$. Defined equivalently by a sequence $F = (F_n)_{n \in \mathbb{N}_0}$ of alternating $n$-linear forms on $X$ with:
\begin{equation*}
    F(v_1 \wedge \cdots \wedge v_n) = F_n(v_1, \ldots, v_n), \hspace{20pt} F(\mathbb{1}_{\bigwedge X}) = F_0 \in \mathbb{R}.
\end{equation*}
We denote by $\mathscr{F}(X)$ the class of fermionic functionals and by $F_n \in \mathscr{F}^{n}(X)$ a fermionic functional of order $n$.
\end{defn}
As the configuration space considered is fixed, we will abbreviate $\mathscr{F}(X)$ simply by $\mathscr{F}$. The space of fermionic functionals $\mathscr{F}$ canonically carries a notion of product. For $F,G \in \mathscr{F}$:
\begin{align}
    (F &\cdot G)_{n}(v_1, \ldots, v_n) \nonumber\\
    &= \sum_{\sigma \in S_n} \mathrm{sign}(\sigma) \sum_{k=0}^{n}\frac{1}{k!(n-k)!} F_k(v_{\sigma(1)}, \ldots, v_{\sigma(k)}) G_{n-k}(v_{\sigma(k+1)}, \ldots, v_{\sigma(n)}) \, , \label{eq: pointwise product}
\end{align}
that is the antisymmetric analogue of the pointwise product. The support of a generic functional $F \in \mathscr{F}$ on $\mathbb{M}$ is defined:
\begin{align*}
    &\mathrm{supp} \, F \coloneqq \\
    &\{ x \in \mathbb{M} | \, \forall \, U \, \mathrm{neighborhood} \,\, \mathrm{of} \,\, x ,\,\, \exists n \in \mathbb{N}, v_1, \ldots, v_n \in X \,\, \mathrm{with} \,\, \mathrm{supp} (v_i) \subset U \,\, \mathrm{s.t.} \,\, F_n(v_1, \ldots, v_n) \neq 0 \}.
\end{align*}
Moreover, as the topological dual of the exterior product of $X$, elements of $\mathscr{F}$ are all compactly supported functionals demanded to be Fréchet differentiable in the following sense
\begin{defn}
Let $F \in \mathscr{F}^n(X)$, $h \in X^{\otimes n -1}$, $\overrightarrow{h} \in X$. The \textbf{left derivative} of $F$ at $h$ in the direction of $\overrightarrow{h}$ is defined for every integer $n \geq 0$
\begin{align*}
    \left\langle \overrightarrow{h} , F^{(1)}(h) \right\rangle &= F\left(\overrightarrow{h} \wedge h\right) \,  \quad \mathrm{for} \,\, n > 0,\\
    F^{(1)} &= 0 \quad F \in \mathscr{F}^{0}(X).
\end{align*}
Instead, the \textbf{right derivative} of $F$ at $h$ in the direction of $\overrightarrow{h}$ is defined for every integer $n \geq 0$
\begin{align*}
    \left\langle F^{(1)}(h), \overrightarrow{h} \right\rangle &= F\left(h \wedge \overrightarrow{h}\right) \,  \quad \mathrm{for} \,\, n > 0,\\
    F^{(1)} &= 0 \quad F \in \mathscr{F}^{0}(X),
\end{align*}
from which:
\begin{equation*}
    \left\langle \overrightarrow{h} , F^{(1)}(h) \right\rangle = (-1)^{n-1} \left\langle F^{(1)}(h), \overrightarrow{h} \right\rangle.
\end{equation*}
The definitions naturally extend to $\mathscr{F}$.
\end{defn}
Let us also introduce a \textbf{degree map} $\mathrm{deg}: \mathscr{F} \to \{0,1\}$ that associates $0$ (respectively $1$) if the order $n$ of the alternating $n$-linear form defining the fermionic functional is even (respectively odd).\\
We restrict our attention to a particular subset of $\mathscr{F}$ denoted $\mathscr{F}_{\mu c}$. This is the set of \textbf{microcausal} fermionic functionals,  over the configuration space $X$:
\begin{equation}\label{eq: microcausal}
    \mathscr{F}_{\mu c} \coloneqq \{ F \in \mathscr{F} | \mathrm{WF}(F^{(n)}) \cap (\overline{V}^{+n} \cup \overline{V}^{-n}) = \emptyset \,\,\,\, \forall n \in \mathbb{N} \},
\end{equation}
where $\mathrm{WF}$ denotes the wave front set \cite{BrouderSmoothIntro}, and $\overline{V}^{\pm} \subset T^* \mathbb{M}$ is the set of all points $(x_1, \ldots, x_d, p^1, \ldots, p^d)$ with all covectors $p^i$ being non-zero, future ($+$) or past ($-$) pointing and either time-like or light-like.\\
Two important subsets of $\mathscr{F}_{\mu c}$ are the local functionals $\mathscr{F}_{loc}$, those for which all functional derivatives are supported on the total diagonal, and the regular functionals $\mathscr{F}_{reg}$, those with smooth functionals as functional derivatives. Finally, the fermionic functionals are endowed with a $\ast$-involution defined for any $F \in \mathscr{F}_{\mu c}$ by:
\begin{equation*}
    F^* \coloneqq F^{\dagger}.
\end{equation*}

Once we endow $\mathscr{F}_{\mu c}$ with the antisymmetric pointwise product defined in \eqref{eq: pointwise product} and the above $\ast$-involution, we obtain the off-shell algebra of classical observables for a formally classical Dirac field propagating on Minkowski spacetime $\mathbb{M}$. The generators are of the form:
\begin{equation}\label{eq: campi lineari}
    \Psi_f(\psi) \coloneqq \int_{\mathbb{M}} \di^4x \, \overline{f}(x) \psi(x)\,, \quad \overline{\Psi}_g\left(\overline{\psi}\right) \coloneqq \int_{\mathbb{M}} \di^4x \, \overline{\psi}(x) g(x) 
\end{equation}
for any $f , g \in \Gamma_0(\mathbb{M}, \mathbb{C}^4)$ and $\psi \in \Gamma(\mathbb{M}, \mathbb{C}^4)$.\\\\

The quantization of the algebra of classical observables is achieved by deforming the product introducing a quantization parameter $\hbar$. To this end, we need to introduce some objects: 
\begin{enumerate}
    \item The map $\mathcal{M}: \mathscr{F}_{\mu c} \otimes \mathscr{F}_{\mu c} \to \mathscr{F}_{\mu c}$ that associates to a tensor product of microcausal functionals their antisymmetric pointwise product \eqref{eq: pointwise product}.
    \item A shortening notation for the left and right Fréchet derivative:
    \begin{equation*}
        \frac{\delta F}{\delta \psi(x)} = F^{(1)}(x)
    \end{equation*}
    and correspondingly for the right derivative $\frac{\delta_r F}{\delta \psi} = (-1)^{\mathrm{deg}(F)} \frac{\delta F}{\delta \psi}$, where $\psi \in \Gamma(\mathbb{M}, \mathbb{C}^4)$ denotes a generic configuration (same notation is used for the Fréchet derivatives with respect to the adjoint configurations $\overline{\psi}$).
    \item The two-point function of a Hadamard state $\omega$. For the sake of completeness, let us recall that Hadamard states are the class of physically relevant quasifree states, see \cite{KW91, Wald1995} for the original paper and a general review, and \cite{FewsterRainer} for precise physical motivations. They have a specific microlocal characterization originally introduced for scalar fields in \cite{Ra96} (see also \cite{BrunettiFredenhagenKohler}) and extended to vector valued fields, further proving that equilibrium states are Hadamard, in \cite{SahlmannVerch2, SahlmannVerch}. Moreover, Hadamard states for the Dirac field are weak bisolutions of the differential operator $\cancel{D} = i \cancel{\partial} - m$, see \cite{Gerardbook,CapoferriMurro, DappiaggiHackPin} for the discussion of Hadamard states for $\cancel{D}$ on globally hyperbolic spacetime, and are connected to the Pauli-Jordan function $\mathscr{S}(x,y)$ by:
    \begin{equation*}
        i\mathscr{S}(x,y) = H^+_{\omega}(x,y) + H^-_{\omega}(x,y),
    \end{equation*}
    where $H^{\pm}_{\omega}(x,y)$ are the two-point functions associated to the Hadamard state $\omega$. Different Hadamard states are distinguished by different choices of smooth functions $R_{\omega} \in \mathcal{C}^{\infty}(\mathbb{M} \times \mathbb{M}, \mathbb{C})$ defined as:
    \begin{align*}
        H^+_{\omega}(x,y) - \frac{i}{2}\mathscr{S}(x,y) &\coloneqq R_{\omega}(x,y)\\
        H^-_{\omega}(x,y) - \frac{i}{2}\mathscr{S}(x,y) &\coloneqq -R_{\omega}(x,y).
    \end{align*}
\end{enumerate} 
Combining everything, choosing a Hadamard state $\omega$, we define for any $F,G \in \mathscr{F}_{\mu c}$ the quantized product:
\begin{equation}\label{eq: DefProd}
    F \star_{\omega} G \coloneqq \mathcal{M}e^{
\hbar \!\!\int\!\! \di^4x \di^4y \, \left( H^+_{\omega}(x,y) \frac{\delta_r}{\delta \psi (x)}
\otimes  \frac{\delta}{\delta \overline{\psi} (y)}
+
H^-_{\omega}(x,y) \frac{\delta_r}{\delta \overline{\psi} (y)}
\otimes  \frac{\delta}{\delta \psi (x)}\right)
 }
F\otimes G,
\end{equation}
that it is well defined by the wave front set properties of the involved distributions \cite{HormanderI, SahlmannVerch}. Moreover, the support of a product of functionals satisfies $\mathrm{supp}(F \star_{\omega} G) \subset \mathrm{supp}(F) \cup \mathrm{supp}(G)$. We denote the quantized algebra by $\mathfrak{A}_{\omega} \coloneqq (\mathscr{F}_{\mu c}, \star_{\omega}, \ast)$ and by $\mathfrak{A}_{\omega}(\mathcal{O})$ the subalgebra of functionals supported in the spacetime region $\mathcal{O} \subset \mathbb{M}$.\\
The algebra $\mathfrak{A}_{\omega}$ is not canonically defined as it depends on the choice of the Hadamard two-point function. Nonetheless, any other $\mathfrak{A}_{\omega'}$ is $\ast$-isomorphic to $\mathfrak{A}_{\omega}$. Indeed, the $\ast$-isomorphism is realized by the map $\alpha_{\omega' - \omega}: \mathfrak{A}_{\omega} \to \mathfrak{A}_{\omega'}$ defined as
\begin{align*}
    &\alpha_{\omega' - \omega}(F) \coloneqq\\
    &\exp\bigg( \frac{1}{2} \int_{\mathbb{M} \times \mathbb{M}} \di^4x \di^4y \, \bigg[ (-1)^n (R_{\omega'} - R_{\omega})(x,y) \frac{\delta^2}{\delta \psi(x) \delta\overline{\psi}(y)} + (-1)^{n+1} (R_{\omega'} - R_{\omega})(x,y) \frac{\delta^2}{\delta\overline{\psi}(y) \delta \psi(x)} \bigg] \bigg) F(\psi, \overline{\psi}).
\end{align*}
The latter, for any $F_1,F_2 \in \mathfrak{A}_{\omega}$ , satisfies:
\begin{equation*}
    \alpha_{\omega' - \omega}(F_1) \star_{\omega'} \alpha_{\omega' - \omega}(F_2) = \alpha_{\omega' - \omega}(F_1 \star_{\omega} F_2)\, , \quad \alpha_{\omega' -\omega}(F_1^*) = \alpha_{\omega' - \omega}(F_1)^*.
\end{equation*}
Therefore, $\mathfrak{A}_{\omega}$ should be interpreted as a realization of the algebra of quantized Dirac fields on Minkowski spacetime whose elements are the normal-ordered functionals with respect to the chosen Hadamard state (see \cite[Theorem $2.6.3.$]{DutschBook} and \cite{HackVerch}). The product automatically implements free anti-commutation relations among linear Dirac fields \eqref{eq: campi lineari}. Finally, as it is customary in the algebraic framework, states are defined as linear, positive and normalized functionals over $\mathfrak{A}_{\omega}$. In particular, the corresponding Hadamard state $\omega$ on $\mathfrak{A}_{\omega}$ associated to the two point functions $H^+_{\omega}, H^-_{\omega}$ consistently gives:
\begin{equation*}
    H^+_{\omega}(x,y) = \omega\big(\Psi(x) \star_{\omega} \overline{\Psi}(y)\big)\, , \quad H^-_{\omega}(x,y) = \omega\big(\overline{\Psi}(y) \star_{\omega} \Psi(x)\big). 
\end{equation*}
In the following we will denote by $\mathfrak{A}$ the general algebra realized in $\mathfrak{A}_{\omega}$ by a corresponding choice of Hadamard state.

\subsection{Free Dynamics and Equilibrium States}
In order to specify equilibrium configurations, we need to define a $*$-dynamical system, namely a $*$-algebra with a one parameter of $*$-automorphisms on it. For a generator $\Psi_{f} \in \mathscr{F}_{\mu c}$, we define $\tau_t(\Psi_{f})$ for $t \in \mathbb{R}$ by its action on the associated kernel:
\begin{equation*}
    \tau_t(\Psi_{f})(x) = \Psi_{f}(x - te^0).
\end{equation*}
and analogously for $\overline{\Psi}_f$. The definition of $\tau_t$ on a generic $F \in \mathcal{F}_{\mu c}$ is obtained demanding $\tau_t$ to be an automorphism on $\mathfrak{A}$. The support of the functionals becomes
\begin{equation*}
    \mathrm{supp}(\tau_t(F)) := \{ x \in \mathbb{M} | (x - t \, e^0) \in \mathrm{supp}(F) \},
\end{equation*}
where $e^0$ is the timelike unit vector in $T\mathbb{M}$ tangent to the time function of Minkowski. Moreover, we demand that it is extended to a $\ast$-automorphism over a realization $\mathfrak{A}_{\omega}$ with respect to a time translation invariant Hadamard state $\omega$ by:
\begin{equation*}
    \tau_t(F \star_{\omega} G) = \tau_t(F) \star_{\omega} \tau_t(G)\, , \quad \forall F,G \in \mathfrak{A}_{\omega}.
\end{equation*}
We call $\tau_t$ \textbf{free dynamics}. It is not strictly necessary to consider $\omega$ time translation invariant. Indeed, if it is not the case the action of the dynamics in a different realization $\mathfrak{A}_{\omega'}$ becomes more cumbersome as it needs to account for the time translation of the two-point functions. Taking this into account, in the following we assume to work in realizations $\mathfrak{A}_{\omega}$ with respect to time translation invariant Hadamard states. Under these assumptions, equilibrium states are defined as follow.
\begin{defn}
Let $\omega^{\infty}$ be a state on a realization of the dynamical algebra $\mathfrak{A}_{\omega}$ with dynamics $\tau_t$. Then, $\omega^{\infty}$ is of equilibrium with respect to $\tau_t$ and it is referred to as ground (vacuum) state if:
\begin{equation}\label{eq: condground}
    -i\partial_t \omega^{\infty}(F^* \star_{\omega} \tau_t(F))|_{t=0} \geq 0 \quad \forall F \in \mathfrak{A}_{\omega}.
\end{equation}
\end{defn}
\begin{defn}\label{def: KMScond}
Let $0 < \beta < \infty$ and $\omega^{\beta}$ be a state on a realization of the dynamical algebra $\mathfrak{A}_{\omega}$ with dynamics $\tau_t$. Then, $\omega^{\beta}$ is a thermal equilibrium or KMS (Kubo-Martin-Schwinger) state, at inverse temperature $\beta$ with respect $\tau_t$, if the functions:
\begin{equation*}
    (t_1, \ldots, t_n) \mapsto \omega^{\beta}(\tau_{t_1}(F_1) \star_{\omega} \cdots \star_{\omega} \tau_{t_n}(F_n)) \, , \quad F_1, \ldots, F_n \in \mathfrak{A}_{\omega}
\end{equation*}
have analytic continuation to the region:
\begin{equation*}
    \{(z_1, \ldots, z_n) \in \mathbb{C}^n \, : \, 0 < \Im(z_j) - \Im(z_i) < \beta, \quad 1 \leq i < j \leq n\},
\end{equation*}
are bounded and continuous on its closure and fulfil the boundary conditions:
\begin{align*}
    \omega^{\beta}(\tau_{t_1}(F_1) \star_{\omega} \cdots \star_{\omega} \tau_{t_{k-1}}(F_{k-1}) \star_{\omega} \tau_{t_{k} + i \beta}(F_{k}) \star_{\omega} \cdots \star_{\omega} \tau_{t_n + i \beta}(F_n))\\
    = \omega^{\beta}(\tau_{t_{k}}(F_{k}) \star_{\omega} \cdots \star_{\omega} \tau_{t_n}(F_n) \star_{\omega} \tau_{t_1}(F_1) \star_{\omega} \cdots \star_{\omega} \tau_{t_{k-1}}(F_{k-1})).
\end{align*}
\end{defn}
The definition of the ground state via condition \eqref{eq: condground} is equivalent to the analiticity of the function $\mathbb{C} \ni z \mapsto \omega^{\infty}(F \tau_z(G))$ for $\Im(z) > 0$ and boundedness for $\Im(z) \geq 0$, see \cite[Proposition $5.3.19$]{BratteliRobinson}.\\
Equivalent characterizations for the KMS condition can also be found in \cite{BratteliRobinson}.

\subsection{Grassmann Algebras, Bosonization and Introduction of Interactions}
As explained in the introduction, in order to discuss interacting Dirac fields we need to deform the introduced fermionic algebra. As the deformation will be performed using Grassmann numbers, we start recalling the definition of a Grassmann algebras.
\begin{defn}
A Grassmann algebra $\mathscr{G}_N$ with $N$ generators $\{ \eta_1, \dots \eta_N\}$, is an algebra, with unit $\mathbb{1}$, an associative product and a $\mathbb{K}$-linear space (with $\mathbb{K}$ a field either $\mathbb{R}$ or $\mathbb{C}$) such that:
\begin{equation*}
    \eta_{\mu} \eta_{\nu} + \eta_{\nu} \eta_{\mu} = 0 \hspace{15pt} \forall \mu,\nu = 1, \dots, N.
\end{equation*}
Any element $\alpha \in \mathscr{G}_N$ can be written as:
\begin{equation} \label{eq: Gdec}
    \alpha = c_0(\alpha) + \sum_{k \geq 1} \sum_{\mu_1, \dots, \mu_k} c_{\mu_1,\dots,\mu_k}(\alpha) \eta_{\mu_1} \cdots \eta_{\mu_k},
\end{equation}
for $c_0(\alpha), c_{\mu_1,\dots,\mu_k}(\alpha) \in \mathbb{K}$.
\end{defn}
Analogously as we did for fermionic functionals a degree map, with a little notation abuse, is introduced $\mathrm{deg}: \mathscr{G}_N \to \{0,1 \}$. With the aid of a Grassmann algebra, we deform the fermionic algebra as follows.
\begin{defn}
Given a real Grassmann algebra $\mathscr{G}_N$ and a realization of the fermionic Algebra $\mathfrak{A}$, we call the correspondingly \textbf{bosonized fermionic algebra}, denoted $\mathfrak{A}^{\mathscr{G}_N} \subset \mathscr{G}_N \otimes \mathfrak{A}$, the unital (unit $\mathbb{1} \otimes \mathbb{1}$) associative algebra generated by elements obtained via the prescription that for any $F \in \mathfrak{A}$:
\begin{equation*}
    \mathfrak{A}^{\mathscr{G}_N} \ni \Tilde{F} \coloneqq \left\{ \begin{aligned}&\eta \otimes F \qquad \mathrm{if} \quad \mathrm{deg}(F) = \mathrm{odd}\\
    &\mathbb{1} \otimes F \qquad \mathrm{if} \quad \mathrm{deg}(F) = \mathrm{even}\end{aligned}\right.
\end{equation*}
where $\eta \in \mathscr{G}_N$ is of odd degree. The associative product on it, for $\alpha_1, \alpha_2 \in \mathscr{G}_N$ and $F_1,F_2 \in \mathfrak{A}$, is given by:
\begin{equation*}
    (\alpha_1 \otimes F_1) \star (\alpha_2 \otimes F_2) = (-1)^{\mathrm{deg}(\alpha_2) \mathrm{deg}(F_1)} (\alpha_1 \alpha_2) \otimes (F_1 \star F_2).
\end{equation*}
Finally, $\mathfrak{A}^{\mathscr{G}_N}$ carries an induced $\ast$-involution:
\begin{equation*}
    \big( (\eta_{\mu_1} \cdots \eta_{\mu_n}) \otimes F \big)^* \coloneqq (\eta_{\mu_n} \cdots \eta_{\mu_1}) \otimes F^*,
\end{equation*}
for all $\mu_{1}, \ldots, \mu_{n} \in \{1, \ldots, N\}$ and $\eta_{\mu_i} \in \mathscr{G}_N$ generators.
\end{defn}
In particular, for the sake of completeness, the support of a bosonized functional is defined to be the support of the corresponding fermionic functional:
\begin{equation}\label{def: bosonization}
    \mathrm{supp}(\eta_{\mu} \otimes F) = \mathrm{supp}(F).
\end{equation}
The elements of $\mathfrak{A}^{\mathscr{G}_N}$ are called \textit{bosonized}, as they now fulfil a Bose statistics. The choice of Grassmann algebra $\mathscr{G}_N$ that was made, will not affect the construction of interacting theories. This non-trivial fact was proven in \cite{BrunettiDuestschRejznerFredenhagenFermions} where the action of $\mathfrak{A}^{\mathscr{G}_N}$ is shown to be functorial in the sense that all operations commute with respect to homomorphisms between finite dimensional Grassmann algebras.\\
As mentioned in the introductory section, the Grassmann degrees of freedom do not have any physical relevance. For this reason, we choose to encode their removal in the definition of state functionals on $\mathfrak{A}^{\mathscr{G}_N}$.
\begin{prop}\label{prop: statiGrass}
Let $\Tilde{\omega}: \mathfrak{A}^{\mathscr{G}_N} \to \mathbb{C}$ be defined for any $\alpha \otimes F \in \mathfrak{A}^{\mathscr{G}_N}$ as:
\begin{equation*}
    \Tilde{\omega}(\alpha \otimes F) \coloneqq \omega^{\mathscr{G}_N}(\alpha) \omega(F)
\end{equation*}
where $\omega: \mathfrak{A} \to \mathbb{C}$ is a state over the fermionic algebra, while $\omega^{\mathscr{G}_N}$ is a functional on the Grassmann algebra defined by the action on generators: 
\begin{align*}
    \omega^{\mathscr{G}_N}(\mathbb{1}) &\coloneqq 1\\
    \omega^{\mathscr{G}_N}(\eta_{\sigma(1)} \cdots \eta_{\sigma(k)}) &= (-1)^{\mathrm{sgn}(\sigma)} \omega^{\mathscr{G}_N}(\eta_1 \cdots \eta_k)\\
    &\coloneqq (-1)^{\mathrm{sgn}(\sigma)}.
\end{align*}
Here, $\sigma$ is any permutation of the $k$ $(\leq N)$ generators and $\mathrm{sgn}$ denotes the sign of the permutation. Then, $\Tilde{\omega}$ is a state on $\mathfrak{A}^{\mathscr{G}_N}$.
\end{prop}
\begin{proof}
Linearity and normalization immediately follow from the given definition. Positivity is a consequence of the fact that, for any $\alpha \in \mathscr{G}_N$ different from the unit and of a specific degree, it holds: $\alpha  \alpha^*= 0$. Therefore:
\begin{equation*}
    \Tilde{\omega}(\alpha\alpha^* \otimes F \star F^*) = 0
\end{equation*}
while, for $F$ of even order we have elements of the form $\mathbb{1} \otimes F$ and therefore:
\begin{equation*}
    \Tilde{\omega}(\mathbb{1} \otimes F \star F^*) \geq 0
\end{equation*}
by $\omega: \mathfrak{A} \to \mathbb{C}$ being a state.
\end{proof}
On $\mathfrak{A}^{\mathscr{G}_N}$ it is possible to discuss interacting fermionic theories. In order to do it, we introduce time ordered products following the original papers \cite{EpsteinGlaser, BrunettiFredenhagen00, HW02, HW05}, see also \cite{DutschBook} for a complete review. We define the time ordered product as a map on local functionals $T_{n,\mathscr{G}_N}:\big(\mathfrak{A}^{\mathscr{G}_N}|_{\mathscr{G}_N \otimes \mathscr{F}_{loc}}\big)^{\otimes n} \to \mathfrak{A}^{\mathscr{G}_N}|_{\mathscr{G}_N \otimes \mathscr{F}_{loc}}$, via the following prescription:
\begin{equation}
    T_{n,\mathscr{G}_N}\big( \alpha_1 \otimes F_1, \ldots, \alpha_n \otimes F_n \big) \coloneqq (\alpha_n \cdots \alpha_1) \otimes T_n\big( F_1, \ldots, F_n \big).
\end{equation}
The time ordered product $T_{n,\mathscr{G}_N}$ satisfies several properties reported in the above cited references. For later purposes, here we only mention some:
\begin{itemize}
    \item Invariance under permutations of $(\alpha_1 \otimes F_1), \ldots, (\alpha_n \otimes F_n) \in \mathscr{G}_N \otimes \mathscr{F}_{loc}$.
    \item The \textit{Causality axiom}, namely for any $(\alpha_1 \otimes F_1), \ldots, (\alpha_n \otimes F_n) \in \mathscr{G}_N \otimes \mathscr{F}_{loc}$:
    \begin{equation}
         T_n\big( F_1, \ldots, F_n \big) = \mathrm{sgn}(\sigma) \, T_k\big( F_{\sigma (1)}, \ldots, F_{\sigma (k)} \big) \star T_{n-k}\big( F_{\sigma (k+1)}, \ldots, F_{\sigma (n)}  \big),
    \end{equation}
    whenever $J^+\big(\mathrm{supp}(F_{\sigma (1)}) \cup \ldots \cup \mathrm{supp}(F_{\sigma (k)})\big) \cap J^{-}\big(\mathrm{supp}(F_{\sigma (k+1)}) \cup \ldots \cup \mathrm{supp}(F_{\sigma (n)})\big) = \emptyset$, $\sigma$ denotes the permutation and $\mathrm{sgn}$ its sign (including the degree of fermionic 
    functionals).
    \item  The representation of the time ordered product as:
    \begin{equation*}
        T_{n,\mathscr{G}_N}\big( \alpha_1 \otimes F_1, \ldots, \alpha_n \otimes F_n \big) = (\alpha_n \cdots \alpha_1) \otimes \big( F_1 \star_{F} \cdots \star_F F_n \big)
    \end{equation*}
    where $\star_F$ is defined as in \eqref{eq: DefProd} replacing the Pauli-Jordan function $\mathscr{S}$ with the fundamental solution of $\cancel{D}$ denoted as $\mathscr{S}_F$ called Feynman propagator of $\cancel{D}$. This representation is unique up to renormalization freedoms, appearing in the extension of powers of $\mathscr{S}_F$.
\end{itemize}
It is at this step that the theory, if needed, is perturbatively renormalized. The time ordered product is inductively constructed so that at each inductive step the possibly involved ill-defined products of distributions are extended non-uniquely to well defined distributions. This is equivalent to renormalization as, the non-unique extensions, are parametrized exactly by the renormalization freedoms \cite{EpsteinGlaser}.\\\\
With the aid of the time ordered product on the bosonized algebra, for any:
\begin{equation*}
    \Tilde{F} = \sum_{j=1}^{N} \lambda_j (a_j \otimes F_j) \, , \qquad a_j \in \{ \mathbb{1}, \eta_j\} \, , \quad F_j \in \mathscr{F}_{loc}\, , \quad \lambda_j \in \mathbb{R} \, , \quad N < \infty
\end{equation*}
the associated $S$-matrix is defined:
\begin{equation}\label{eq: S matrix}
    S(\Tilde{F}) \coloneqq \mathbb{1} + \sum_{n=1}^{\infty} \frac{i^n}{n! \hbar^n}T_{n, \mathscr{G}_N}\big( \Tilde{F}^{\otimes n} \big)
\end{equation}
where $j_1, \ldots, j_n \in \{1, \ldots, N\}$. The necessity of the bosonization manifests as now $S$-matrices are generators of $n$-th order time ordered products:
\begin{equation*}
    (a_{j_1} \cdots a_{j_n}) \otimes T_n(F_{j_1} \otimes \cdots \otimes F_{j_n}) = \frac{\hbar^n}{i^n}\frac{\partial^n}{\partial \lambda_{j_1} \cdots \partial \lambda_{j_n}}\bigg|_{\lambda_1 = \ldots = \lambda_{N} = 0} S(\Tilde{F}).
\end{equation*}
The so defined $S$-matrix has inverse $S(\Tilde{F})^{-1} = S(\Tilde{F}^*)^*$ and satisfies, as a consequence of the analogous property for time ordered products, the causal factorization property:
\begin{equation}\label{eq: Causal Fact}
    S(\Tilde{F}_{1} + \Tilde{F}_{2} + \Tilde{F}_{3}) = S(\Tilde{F}_{1} + \Tilde{F}_{2}) \star S(\Tilde{F}_{2})^{\star-1}\star S(\Tilde{F}_{2} + \Tilde{F}_{3})
\end{equation}
for $\Tilde{F}_{1}, \Tilde{F}_{2}, \Tilde{F}_{3}$ even functionals such that $\mathrm{supp} \Tilde{F}_{1} \cap J^{-}(\mathrm{supp} \Tilde{F}_{3}) = \emptyset$.\\\\

The interactions considered in this paper are elements in $\mathfrak{A}^{\mathscr{G}_N}$, for $\chi \in \mathcal{C}^{\infty}_0(\mathbb{R})$ and $h \in \mathcal{C}_0^{\infty}(\Sigma)$, defined by the potential:
\begin{equation}\label{eq: formaV}
   \Tilde{V} \coloneqq \mathbb{1} \otimes V = \mathbb{1} \otimes \lambda \bigg( \int_{\mathbb{M}} \chi(t) h(\mathbf{x}) \mathcal{L}_I \di^4x \bigg),
\end{equation}
where it is assumed that $\mathcal{L}_I$ is invariant under spacetime translation, up to a possible smooth bounded function. Namely, $\mathcal{L}_{I} = (\overline{\psi} \psi)^n$ or $\mathcal{L}_I = (\overline{\psi} \cancel{P}(\mathbf{x}) \psi)^{n}$ for $n \in \mathbb{N}$ and $P: \mathbb{R}^3 \to \mathbb{R}$ so that $V$ is a microcausal functional. We notice that in this case, juxtaposition by any Grassmann variable is not needed as, in order for the full Lagrangian to inherit the continuous and discrete symmetries of the free Lagrangian, $\mathcal{L}_I$ must be of even degree.\\
Accordingly, interacting observables for an interacting Dirac theory with interaction $V$ are defined. Namely, in the sense of formal power series, exists a map (called \textit{Bogoliubov map}) $R_{\Tilde{V}}: \mathfrak{A}^{\mathscr{G}_N} \to \mathfrak{A}^{\mathscr{G}_N}[[\lambda]]$ that maps observables of the interacting theory into formal power series in the coupling constant $\lambda$ in $\mathfrak{A}^{\mathscr{G}_N}$. The Bogoliubov map has the following form:
\begin{equation}\label{eq: Bogol}
    R_{\Tilde{V}}(\Tilde{F}) \coloneqq -\frac{i}{\lambda} \frac{d}{ds}\bigg|_{s=0} S(\Tilde{V})^{-1} \star S(\Tilde{V} + s \Tilde{F}) = S(\Tilde{V})^{-1} \star (T_{\mathscr{G}_N}(\Tilde{F}, S(\Tilde{V}))),
\end{equation}
where again, by the evenness of the interaction, Grassmann variables enter just in the possible bosonization of the observable $\Tilde{F}$. The set of all interacting observables is a subalgebra of $\mathfrak{A}^{\mathscr{G}_N}$ defined as:
\begin{equation*}
    \mathfrak{A}_I^{\mathscr{G}_N}(\mathcal{O}) \coloneqq \big[ \{ R_{\Tilde{V}}(\Tilde{F}) | \mathrm{supp}(\Tilde{F}) \subset \mathcal{O} \} \big],
\end{equation*}
where the square brackets denote the algebra generated by the elements in them.
Notice that, by the causal factorization property, $\mathfrak{A}_I^{\mathscr{G}_N}(\mathcal{O})$ is not sensitive to changes of $\chi$ to $\chi'$ when $\text{supp} \big((\chi-\chi')h \big)\cap J^{-}(\mathcal{O}) = \emptyset$. As a consequence, as long as the observables are supported after the complete switch on, we may extend $\chi$ to:
\begin{equation}\label{eq:accad}
    \chi(t) = \left\{ \begin{aligned}
                          &0 \qquad \text{for $t < -2\epsilon$}\\
                          &1 \qquad \text{for $t \geq  -\epsilon$}
                      \end{aligned}\right.
\end{equation}
for any $\epsilon \in \mathbb{R}^+$. In what follows, we refer to this as the \textbf{smooth switch-on function}.\\\\
Finally, as the aim is constructing equilibrium states for interacting theories, we define (perturbatively) the dynamics on $\mathfrak{A}_I^{\mathscr{G}_N}$ induced by the full Lagrangian. The latter is defined as the pullback of the free dynamics with respect to the Bogoliubov map as:
\begin{equation}\label{def: InterDynam}
    \tau_t^{\Tilde{V}} R_{\Tilde{V}}(\Tilde{F}) \coloneqq R_{\Tilde{V}}(\tau _t \Tilde{F})\, , \quad \forall t \in \mathbb{R},
\end{equation}
where $\tau_t$ is the extension of the free dynamics on $\mathfrak{A}^{\mathscr{G}_N}$. The given definition of $\tau_t^{\Tilde{V}}$, makes it an automorphism on $\mathfrak{A}_I^{\mathscr{G}_N}$ describing the interacting dynamics. The authors of \cite{FredenhagenLindnerKMS_2014} noticed that the defined automorphism in \eqref{def: InterDynam} is equivalent to the existence of a unitary $U_V(t)$ for $t \in \mathbb{R}$ intertwining the free and the interacting dynamics (Theorem $1$ of \cite{FredenhagenLindnerKMS_2014}). The first crucial point is to avoid the infrared divergences introducing the space cutoff $h \in \mathcal{C}^{\infty}_0(\mathbb{R}^3)$ and time switching function $\chi$ and discuss their dependence and removal at the level of expectation values. The second point is to avoid the ultraviolet divergences by the aid of the time-slice axiom. The latter states that the algebra of observables localized within an arbitrarily small time interval suffice to determine any other observable of the theory with arbitrary localization up to terms vanishing on-shell. Namely, denoting $\Sigma_{\epsilon} \coloneqq \{ (t,\mathbf{x}) : -\epsilon < t < \epsilon \}$, for any Hadamard state $\omega$ and $F \in \mathfrak{A}$ observable for the scalar theory:
\begin{equation*}
    \omega\left(R_{V}(F)\right) = \omega\left(R_{V}(F_{\epsilon})\right)
\end{equation*}
for some $F_{\epsilon} \in \mathfrak{A}(\mathcal{O})$ with $\mathcal{O} \subset \Sigma_{\epsilon}$. This was proven to hold for interacting scalar fields, in this formal renormalized perturbative framework, on globally hyperbolic spacetimes in \cite{ChilianFredenhagen} allowing to avoid divergences arising from the restriction to an initial data Cauchy surface, as discussed in the introduction.\\

The cited work focuses on interacting scalar QFT, but the just mentioned statement directly generalizes to our case. Indeed, following the proof in \cite{ChilianFredenhagen}, the time-slice axiom holds also in the fermionic setting by mean of the performed twist on the algebraic structure ensuring the validity of the causal factorization property of $S$-matrices \eqref{eq: Causal Fact} and their normalization $S(0) = \mathbb{1}$.\\

As a consequence, there exist a $U_{\Tilde{V}}(t) \in \mathfrak{A}^{\mathscr{G}_N}_I$ unitary for $t \in \mathbb{R}$ that satisfies:
\begin{equation*}
    \tau_t^{\Tilde{V}}(\Tilde{F})= U_{\Tilde{V}}(t) \star \tau_t(\Tilde{F}) \star U_{\Tilde{V}}(t)^{-1}.
\end{equation*}
Such a unitary, satisfies the cocycle condition $U_{\Tilde{V}}(t+s) = U_{\Tilde{V}}(t) \star \tau_t (U_{\Tilde{V}}(s))$ and has a perturbative definition in terms of a formal power series:
\begin{equation*}
    U_{\Tilde{V}}(t) = \mathbb{1} + \sum_{n=1}^{\infty} i^n \int_0^t \di t_1 \int_0^{t_1} \di t_2 \cdots \int_0^{t_{n-1}} \di t_n \tau_{t_n}(\Tilde{K}) \star \cdots \star \tau_{t_1}(\Tilde{K})
\end{equation*}
in its generator:
\begin{equation}\label{eq: generatore}
    \Tilde{K} = - R_{\Tilde{V}}\bigg( \mathbb{1} \otimes \bigg(\lambda \int_{\mathbb{M}} \di^4x \, h(\mathbf{x}) \Dot{\chi}(t) \mathcal{L}_I \bigg) \bigg),
\end{equation}
where $\Dot{\chi}(t)$ denotes the derivative of the switch-on function $\chi(t)$. As a consequence, the interacting dynamics itself admits a perturbative expansion in terms of the free one:
\begin{equation*}
    \tau_t^{\Tilde{V}}(\Tilde{F}) = \tau_t(\Tilde{F}) + \sum_{n\geq 1}i^n 
\int_{t S_n}[ \tau_{t_1}(\Tilde{K}),[\dots ,[\tau_{t_n}(\Tilde{K}),\tau_t(\Tilde{F}) ]\dots  ] ]_{\star}
\di t_1 \dots \di t_n,
\end{equation*}

\section{Equilibrium States for Interacting Fermi Fields}\label{sec: SuppComp}
In this section we present the construction of KMS states for interacting Fermi fields. Here, following the ideas in \cite{FredenhagenLindnerKMS_2014} with all technical distinctions discussed, we present the construction for spatially compact interactions. The removal of the space cutoff, $h \to 1$, is devoted to the next section.
\begin{thm}\label{thm: stati}
Let $\Tilde{\omega}^{\beta}$ be the extension to $\mathfrak{A}^{\mathscr{G}_N}$, the bosonized algebra of smeared Wick polynomials, of a KMS state (inverse temperature $0 < \beta < \infty$) or a ground state ($\beta = \infty$) with respect to $\tau_t$ on the fermionic algebra with two point functions:
\begin{align}
    \omega_2^{\beta,\pm}(x,y) &= \frac{1}{(2 \pi)^3} \int \frac{\di^3\mathbf{p}}{2 \omega_p} \bigg( \frac{(-\gamma^0 \omega_p - \gamma^i p_i + m)e^{-i\omega_p (t_x - t_y)}}{(1 + e^{\mp\beta \omega_p})} - \frac{(\gamma^0 \omega_p - \gamma^i p_i + m)e^{i\omega_p (t_x - t_y)}}{(1 + e^{\pm\beta \omega_p})}  \bigg) e^{i \mathbf{p} (\mathbf{x}- \mathbf{y})} \label{eq: 2puntiKMS}\\
    \omega_2^{\infty,\pm}(x,y) &= \frac{1}{(2 \pi)^3} \int \frac{\di^3\mathbf{p}}{2 \omega_p} (-\gamma^0 \omega_p - \gamma^i p_i \pm m)e^{\mp i\left(\omega_p (t_x - t_y) - \mathbf{p} (\mathbf{x}- \mathbf{y})\right)} \label{eq: 2puntiGround}.
\end{align}
Then, the following holds in the sense of formal power series:
\begin{itemize}
    \item For $\Tilde{F}_1, \ldots, \Tilde{F}_n \in \mathfrak{A}^{\mathscr{G}_N}(\mathcal{O})$ for $\mathcal{O} \subset \Sigma_{\epsilon}$ and $r,t_1, \ldots, t_n,s \in \mathbb{R}$ the function:
    \begin{equation*}
        G_{\Tilde{F}_1, \cdots, \Tilde{F}_n}(r,t_1, \ldots, t_n,s) = \frac{\Tilde{\omega}^{\beta}(U_{\Tilde{V}}(r)^{-1} \star \tau_{t_1}^{\Tilde{V}}(\Tilde{F}_1) \star \cdots \star \tau_{t_n}^{\Tilde{V}}(\Tilde{F}_n) \star U_{\Tilde{V}}(s))}{\Tilde{\omega}^{\beta}(U_{\Tilde{V}}(s-r))}
    \end{equation*}
    can be extended to a continuous function on the closure, and is analytic in the interior, of the strip:
    \begin{equation*}
        \mathfrak{I}_{n+2}^{\beta} \coloneqq \{ (z_1, \ldots, z_{n+2}) \in \mathbb{C}^n : 0 < \big|\Im(z_i) - \Im(z_j)\big| < \beta, \quad 1 \leq i < j \leq n +2 \}.
    \end{equation*}
    \item For $0 < \beta < \infty$ the linear functional $\Tilde{F} \mapsto \Tilde{\omega}^{\beta, \Tilde{V}}(\Tilde{F}) \coloneqq G_{\Tilde{F}}(-i\beta /2, 0, i\beta /2)$ is a KMS state with respect to $\tau_t^{\Tilde{V}}$
    \item If for $\beta = \infty$ the limit $\lim_{\beta' \to \infty} G_{\Tilde{F}_1, \cdots, \Tilde{F}_n}(-i\beta'/2, \ldots, i\beta'/2)$ exists uniformly on any compact set in $\mathfrak{I}_{n}^{\infty}$ for all $\Tilde{F}_1, \ldots, \Tilde{F}_n \in \mathfrak{A}^{\mathscr{G}_N}((\mathcal{O}))$, then $\Tilde{\omega}^{\infty, \Tilde{V}}(\Tilde{F}) = \lim_{\beta' \to \infty} G_{\Tilde{F}}(-i\beta'/2, 0, i\beta'/2)$ is a ground state with respect to $\tau_t^{\Tilde{V}}$
\end{itemize}
\end{thm}
\begin{proof}
For the first statement, we need to study the analytic continuation of:
\begin{equation*}
    G_n(r,t_1, \ldots, t_n,s) = \Tilde{\omega}^{\beta}(U_{\Tilde{V}}(r)^{-1} \star \tau_{t_1}^{\Tilde{V}}(\Tilde{F}_1) \star \cdots \star \tau_{t_n}^{\Tilde{V}}(\Tilde{F}_n) \star U_{\Tilde{V}}(s))
\end{equation*}
in the sense of formal power series. Therefore, consider the $l$-th order expansion in $\Tilde{K}$ arising from the presence of both the unitary cocycle and the interacting dynamics. In particular, by the relation between the interacting and the free dynamics:
\begin{align*}
     U_{\Tilde{V}}(t,s) &= U_{\Tilde{V}}(t)^{-1} \star U_{\Tilde{V}}(s)\\
     &= \sum_{n=0}^{\infty} (i(s-t))^n \int_{\mathcal{S}_n} \di^nu \, \tau_{t+ u_1(s-t)}(\Tilde{K}) \star \cdots \star \tau_{t+ u_n(s-t)}(\Tilde{K}),
\end{align*}
where we have introduced the unit simplex:
\begin{equation*}
    \mathcal{S}_n \coloneqq \{ (u_1, \ldots, u_n) \in \mathbb{R}^n : 0 \leq u_1 \leq \ldots \leq u_n \leq 1 \}.
\end{equation*}
Follows that at the $l$-th order we have:
\begin{align*}
    &G_n^{(l)}(r, t_1, \ldots, t_n, s) = \sum_{\mathbf{m} \in \mathbb{N}^{n+2}\, , \, |\mathbf{m}| = l} \int_{\mathcal{S}_{m_1}} \di^{m_1}u^{(1)} \cdots \int_{\mathcal{S}_{m_{n+2}}} \di^{m_{n+2}}u^{(n+2)}(i(t_1 - r))^{m_{1}} \cdots (i(s - t_n))^{m_{n+2}}\\
    &\times \Tilde{\omega}^{\beta}\bigg(\prod_{j=1}^{m_{1}}\bigg[ \tau_{r + u^{(1)}_j(t_1-r)}(\Tilde{K}) \bigg] \star \tau_{t_1}(\Tilde{F_1}) \star \prod_{j=1}^{m_{2}}\bigg[ \tau_{t_1 + u^{(2)}_j(t_2-t_1)}(\Tilde{K}) \bigg] \star \cdots \star \prod_{j=1}^{m_{n+2}}\bigg[ \tau_{t_n + u^{(n+2)}_j(s-t_n)}(\Tilde{K}) \bigg]\bigg)
\end{align*}
where $\mathbf{m}$ is a multiindex of order $l$. Now, consider the complex extension of $(r, t_1, \ldots, t_n, s)$ to $(z_1, \ldots, z_{n+2}) \in \mathfrak{I}_{n+2}^{\beta}$. We want to show that $G_n^{(l)}(z_1, \ldots, z_{n+2})$ is analytic in the interior and continuous on $\mathfrak{I}_{n+2}^{\beta}$. For that purpose, we need to show:
\begin{align*}
    &0 < \big|\Im\big(z_i + u^{(k)}_{j}(z_{i+1} - z_i)\big) - \Im\big(z_i + u^{(k)}_{j+1}(z_{i+1} - z_i)\big) \big|< \beta\\
    &0 < \big|\Im\big(z_i + u^{(k')}_{m_{k}}(z_{i+1} - z_i)\big) - \Im\big(  z_{i+1}\big) \big|< \beta,
\end{align*}
where $i = 1, \ldots, n+1$, $k = m_1, \ldots, m_{n+2}$, $k' = m_1, \ldots, m_{n+1}$. In this way, the analiticity and continuity properties of the restriction of the state on the fermionic algebra, $\omega^{\beta}$, give the same desired properties for $G_n^{(l)}(z_1, \ldots, z_{n+2})$.\\
We start from the first and notice:
\begin{equation*}
    \Im\big(z_i + u^{(k)}_{j}(z_{i+1} - z_i)\big) - \Im\big(z_i + u^{(k)}_{j+1}(z_{i+1} - z_i)\big) = \Im\big(z_i - z_{i+1}\big)(u_{j+1}^{(k)} - u_j^{(k)}).
\end{equation*}
But now, since $(z_1, \ldots, z_{n+2}) \in \mathfrak{I}_{n+2}^{\beta}$ and $(u^{(k)}_{1}, \ldots, u^{(k)}_{m_k}) \in \mathcal{S}_{k}$ the claim follows. The second condition is equivalent to:
\begin{equation*}
    \Im\big(z_i + u^{(k')}_{m_{k}}(z_{i+1} - z_i)\big) - \Im\big(  z_{i+1}\big) = \Im\big( z_i - z_{i+1} \big)(1 - u^{(k')}_{m_{k}}).
\end{equation*}
Again using $(z_1, \ldots, z_{n+2}) \in \mathfrak{I}_{n+2}^{\beta}$ and $(u^{(k)}_{1}, \ldots, u^{(k)}_{m_k}) \in \mathcal{S}_{k}$ the claim follows. Therefore by the continuity and analiticity properties of $\omega^{\beta}$ the claim follows\\
For the second statement, we start setting $r = -i\beta/2$ and $s = i\beta/2$. We want to show that, at any perturbative order $l$ in the expansion:
\begin{equation*}
    G_n^{(l)}(-i\beta/2, t_1, \ldots, t_{k-1}, t_k + i\beta, \ldots, t_n + i\beta, i\beta/2) = G_n^{(l)}(-i\beta/2, t_k, \ldots, t_{n}, t_1, \ldots, t_{k-1}, i\beta/2).
\end{equation*}
Expand the left hand side:
\begin{align*}
    &G_n^{(l)}(-i\beta/2, t_1, \ldots, t_{k-1}, t_k + i\beta, \ldots, t_n + i\beta, i\beta/2)\\
    &= \sum_{\mathbf{m} \in \mathbb{N}^{n+2}\, , \, |\mathbf{m}| = l} \int_{\mathcal{S}_{m_1}} \di^{m_1}u^{(1)} \cdots \int_{\mathcal{S}_{m_{n+2}}} \di^{m_{n+2}}u^{(n+2)}(i(t_1 + i\beta/2))^{m_{1}} \cdots (i(i\beta/2 - t_n))^{m_{n+2}}\\
    &\times \Tilde{\omega}^{\beta}\bigg(\prod_{j=1}^{m_{1}}\bigg[ \tau_{-i\beta/2 + u^{(1)}_j(t_1+i\beta/2)}(\Tilde{K}) \bigg] \star \tau_{t_1}(\Tilde{F_1}) \star \prod_{j=1}^{m_{2}}\bigg[ \tau_{t_1 + u^{(2)}_j(t_2-t_1)}(\Tilde{K}) \bigg] \star \cdots \star \prod_{j=1}^{m_{k}}\bigg[ \tau_{t_{k-1} + u^{(2)}_j(t_k + i \beta -t_{k-1})}(\Tilde{K}) \bigg]\\ &\star \tau_{t_k + i\beta}(\Tilde{F_k}) \star \prod_{j=1}^{m_{k+1}}\bigg[ \tau_{t_{k} + i\beta + u^{(2)}_j(t_{k+1}-t_k)}(\Tilde{K}) \bigg] \star \cdots \star \prod_{j=1}^{m_{n+2}}\bigg[ \tau_{t_n + i\beta + u^{(n+2)}_j(-i\beta/2-t_n)}(\Tilde{K}) \bigg]\bigg)\\
    &= \sum_{\mathbf{m} \in \mathbb{N}^{n+2}\, , \, |\mathbf{m}| = l} \int_{\mathcal{S}_{m_1}} \di^{m_1}u^{(1)} \cdots \int_{\mathcal{S}_{m_{n+2}}} \di^{m_{n+2}}u^{(n+2)}(i(t_1 + i\beta/2))^{m_{1}} \cdots (i(i\beta/2 - t_n))^{m_{n+2}}\\
    &\times \Tilde{\omega}^{\beta}\bigg(\tau_{t_k}(\Tilde{F_k})
    \star \prod_{j=1}^{m_{k+1}}\bigg[ \tau_{t_{k} + u^{(2)}_j(t_{k+1}-t_k)}(\Tilde{K}) \bigg] \star \cdots \star \prod_{j=1}^{m_{n+2}}\bigg[ \tau_{t_n + u^{(n+2)}_j(-i\beta/2-t_n)}(\Tilde{K}) \bigg]\\
    & \star \prod_{j=1}^{m_{1}}\bigg[ \tau_{-i\beta/2 + u^{(1)}_j(t_1+i\beta/2)}(\Tilde{K}) \bigg]
    \star \tau_{t_1}(\Tilde{F_1}) \star \prod_{j=1}^{m_{2}}\bigg[ \tau_{t_1 + u^{(2)}_j(t_2-t_1)}(\Tilde{K}) \bigg] \star \cdots \star \prod_{j=1}^{m_{k}}\bigg[ \tau_{t_{k-1} + u^{(2)}_j(t_k + i \beta -t_{k-1})}(\Tilde{K}) \bigg]\bigg)
\end{align*}
where we have used the definition of $\Tilde{\omega}^{\beta}$ given in Proposition \ref{prop: statiGrass} and the KMS condition for $\omega^{\beta}$ with respect to the free dynamics. But now the last term is nothing but the contribution coming from:
\begin{equation*}
    U_{\Tilde{V}}(t_{k-1})^{-1} \star U_{\Tilde{V}}(t_k + i\beta) = U_{\Tilde{V}}(t_{k-1})^{-1} \star U_{\Tilde{V}}(i\beta/2) \star \tau_{i\beta}\big( U_{\Tilde{V}}(-i\beta/2)^{-1} \star U_{\Tilde{V}}(t_k)  \big).
\end{equation*}
While the two adjacent products of terms invoving just $\Tilde{K}$ comes from the expansion of:
\begin{equation*}
    U_{\Tilde{V}}(t_n)^{-1} \star U_{\Tilde{V}}(-i\beta/2) \star U_{\Tilde{V}}(-i\beta/2)^{-1} \star U_{\Tilde{V}}(t_1) = U_{\Tilde{V}}(t_n)^{-1} \star U_{\Tilde{V}}(t_1)
\end{equation*}
Therefore, using again the KMS condition on the last factor, we get the desired equality and thus the KMS property in the sense of formal power series with respect to the interacting dynamics $\tau_t^{\Tilde{V}}$. Finally, for this to be a state, it remains to be proven the positivity in the sense of formal power series. Taking into account that $U_{\Tilde{V}}(i\beta/2)^* = U_{\Tilde{V}}(-i\beta/2)^{-1}$, we get for any $A \in \mathfrak{A}^{\mathscr{G}_N}$:
\begin{align*}
    \Tilde{\omega}^{\beta}\big(U_{\Tilde{V}}(-i\beta/2)^{-1} \star A^* \star A \star U_{\Tilde{V}}(i\beta/2)\big) &= \Tilde{\omega}^{\beta}\big(U_{\Tilde{V}}(i\beta/2)^* \star A^* \star A \star U_{\Tilde{V}}(i\beta/2)\big)\\
    &= \Tilde{\omega}^{\beta}\big( B^* \star B \big)
\end{align*}
where we have introduced $B = A \star U_{\Tilde{V}}(i\beta/2)$. But now the above is positive in the sense of formal power series by the positivity of the state $\Tilde{\omega}^{\beta}$ (see Proposition \ref{prop: statiGrass}).\\\\
Finally we focus on the construction of the interacting ground state. We need to show that on any compact $\mathcal{K} \subset \mathfrak{I}_n^{\infty}$ and at each perturbative order:
\begin{equation*}
    \sup_{\mathcal{K} \subset \mathfrak{I}_n^{\infty}}\lim_{\beta'\to \infty} \bigg|\left( \Tilde{\omega}^{\beta'} - \Tilde{\omega}^{\infty}\right)\left(U_{\Tilde{V}}\left(-\frac{i\beta'}{2}\right)^{-1} \star \tau_{t_1}^{\Tilde{V}}(\Tilde{F}_1) \star \cdots \star \tau_{t_n}^{\Tilde{V}}(\Tilde{F}_n) \star U_{\Tilde{V}}\left(\frac{i\beta'}{2}\right)\right)\bigg| = 0
\end{equation*}
Namely:
\begin{align*}
    &\sup_{\mathcal{K} \subset \mathfrak{I}_n^{\infty}} \lim_{\beta'\to \infty} \bigg| \big(\Tilde{\omega}^{\beta'} - \Tilde{\omega}^{\infty}\big) \bigg(\sum_{\mathbf{m} \in \mathbb{N}^{n+2}\, , \, |\mathbf{m}| = l} \int_{\mathcal{S}_{m_1}} \di^{m_1}u^{(1)} \cdots \int_{\mathcal{S}_{m_{n+2}}} \di^{m_{n+2}}u^{(n+2)}(i(t_1 + i\beta'/2))^{m_{1}} \cdots (i(i\beta'/2 - t_n))^{m_{n+2}}\\
    &\times \prod_{j=1}^{m_{1}}\bigg[ \tau_{-i\beta'/2 + u^{(1)}_j(t_1+i\beta'/2)}(\Tilde{K}) \bigg] \star \tau_{t_1}(\Tilde{F_1}) \star \prod_{j=1}^{m_{2}}\bigg[ \tau_{t_1 + u^{(2)}_j(t_2-t_1)}(\Tilde{K}) \bigg] \star \cdots \star \prod_{j=1}^{m_{n+2}}\bigg[ \tau_{t_n + u^{(n+2)}_j(-i\beta'/2-t_n)}(\Tilde{K}) \bigg]\bigg) \bigg| = 0.
\end{align*}
However, both $\omega^{\beta'}$ and $\omega^{\infty}$ are Hadamard states. Therefore, since the part of the states acting on the Grassmann variables is universal, the difference of the corresponding two point functions is a smooth function. For our case of interest, fermionic fields on Minkowski spacetime, it has the form:
\begin{align*}
    W^{\beta'}(x-y) &= |\big(\omega^{\beta'}_2 - \omega^{\infty}_2\big)(x-y)|\\
    &= \bigg| \frac{1}{(2\pi)^3} \int \frac{d^3\mathbf{p}}{2 \omega_p} \frac{e^{i \mathbf{p} \cdot (\mathbf{x} - \mathbf{y})}}{1 + e^{\beta' \omega_p}}\big((-\gamma^0 \omega_p - \gamma^i p_i +m)e^{-i\omega_p(t_x - t_y)} +  (\gamma^0 \omega_p - \gamma^i p_i +m)e^{i\omega_p(t_x - t_y)} \big) \bigg|
\end{align*}
Follows that, by Wick's theorem, the above expression is just a product of such smooth function. However, in the limit $\beta' \to \infty$ all of them tend to zero independently on any compact $\mathcal{K}$ due to the presence of the Fermi factors. 
Thus, everything converges to zero as also the additional growths are all polynomial (the terms $(i\beta'/2 + t_1)^{m_1}$). The claim is proven as consequence of the shown analiticity properties of the limiting state.
\end{proof}
\begin{rem}\label{rem: 2}
By the KMS condition, the equilibrium state at finite $\beta$ for the interacting theory $\Tilde{\omega}^{\beta,V}$ can also be written for any $\Tilde{A} \in \mathfrak{A}^{\mathscr{G}_N}(\mathcal{O})$ as:
\begin{equation*}
    \Tilde{\omega}^{\beta,V}(\Tilde{A}) = \frac{\Tilde{\omega}^{\beta}\big(\Tilde{A} \star U_{\Tilde{V}}(i\beta)\big)}{\Tilde{\omega}^{\beta}\big( U_{\Tilde{V}}(i\beta)\big)}.
\end{equation*}
\end{rem}

Moreover, as already pointed out in the proof of the above theorem, discussions regarding properties of states on the full bosonized algebra are equivalently performed focusing only on the part acting on the fermionic algebra by the universality of the action on Grassmann variables. For this reason, keeping it in mind, all the incoming discussions focus only on states on the fermionic algebra.\\

For later purposes, we recall that any state admits an expansion in its connected (truncated) parts \cite{BratelliKishimotoRobinson}. For any $\omega$ on $\mathfrak{A}$ and $A_1, \ldots, A_n \in \mathfrak{A}$:
\begin{equation*}
    \omega(A_1, \ldots, A_n) = \sum_{P \in \mathrm{Part}\{1, \ldots, n\}} \prod_{I \in P} \omega^{\mathcal{T}}\bigg( \bigotimes_{i \in I} A_i \bigg).
\end{equation*}
Here $\omega^{\mathcal{T}}$ denotes the connected correlation functions defined as functionals over the tensor algebra constructed over $\mathfrak{A}$ and $\mathrm{Part}\{1, \ldots, n\}$ the set of all partitions of $\{1, \ldots, n\}$ into non-void subsets.\\
As a consequence, Proposition $3$ in \cite{FredenhagenLindnerKMS_2014}, gives the perturbative expansion of the interacting KMS state in terms of connected correlation functions:
\begin{equation}\label{eq: espansionestato}
    \omega^{\beta,V}(A) = \sum_{n=0}^{\infty} (-1)^n \int_{\mathcal{S}_{n,\beta}} \di^n u \,\, \omega^{\beta,\mathcal{T}} \bigg( \bigotimes_{j: u_j < 0} \tau_{iu_j}(K) \otimes A \otimes \bigotimes_{j: u_j \geq 0}\tau_{iu_j}(K) \bigg)\, ,
\end{equation}
where $\omega^{\beta,V}$ denotes the KMS state for the interacting theory with spatial cutoff function $h$ and:
\begin{equation*}
    \mathcal{S}_{n,\beta} \coloneqq \{ -\beta/2 \leq u_1 \leq \ldots \leq u_n \leq \beta/2 \}
\end{equation*}

\section{Adiabatic Limit}\label{sec: LimAd}
We discuss the removing of the spacetime cutoff: the limit $h \to 1$ for the states constructed in the previous section. Worth to mention is that the adiabatic limit, for fermionic theories, can be successfully taken also in the case of massless theories.\\
The presence of the $U_{\Tilde{V}}(i \beta/2)$ term in the definition of the interacting equilibrium state, makes necessary a systematic discussion of the existence of the limit $h \to 1$. Indeed, the Bogoliubov map itself does not give issues with such a limit. In fact, a state of the interacting theory naively constructed from a free state $\Tilde{\omega}$ via the prescription:
\begin{equation*}
    \Tilde{\omega}^I(\Tilde{A}) = \Tilde{\omega}(R_{\Tilde{V}}(\Tilde{A})) \quad \forall \Tilde{A} \in \mathfrak{A}^{\mathscr{G}_N},
\end{equation*}
has always a convergent $h \to 1$ limit as a consequence of the causal dependence of the Bogoliubov map. Indeed, $R_{\Tilde{V}}(\Tilde{A})$ is supported in the region $J^-(\mathrm{supp}(\Tilde{A})) \cap \mathrm{supp}(\Tilde{V})$ that is always compact as illustrated in Figure \ref{fig:1}.
\begin{figure}[H]
 		\centering
 		\includegraphics[width=1.00\columnwidth]{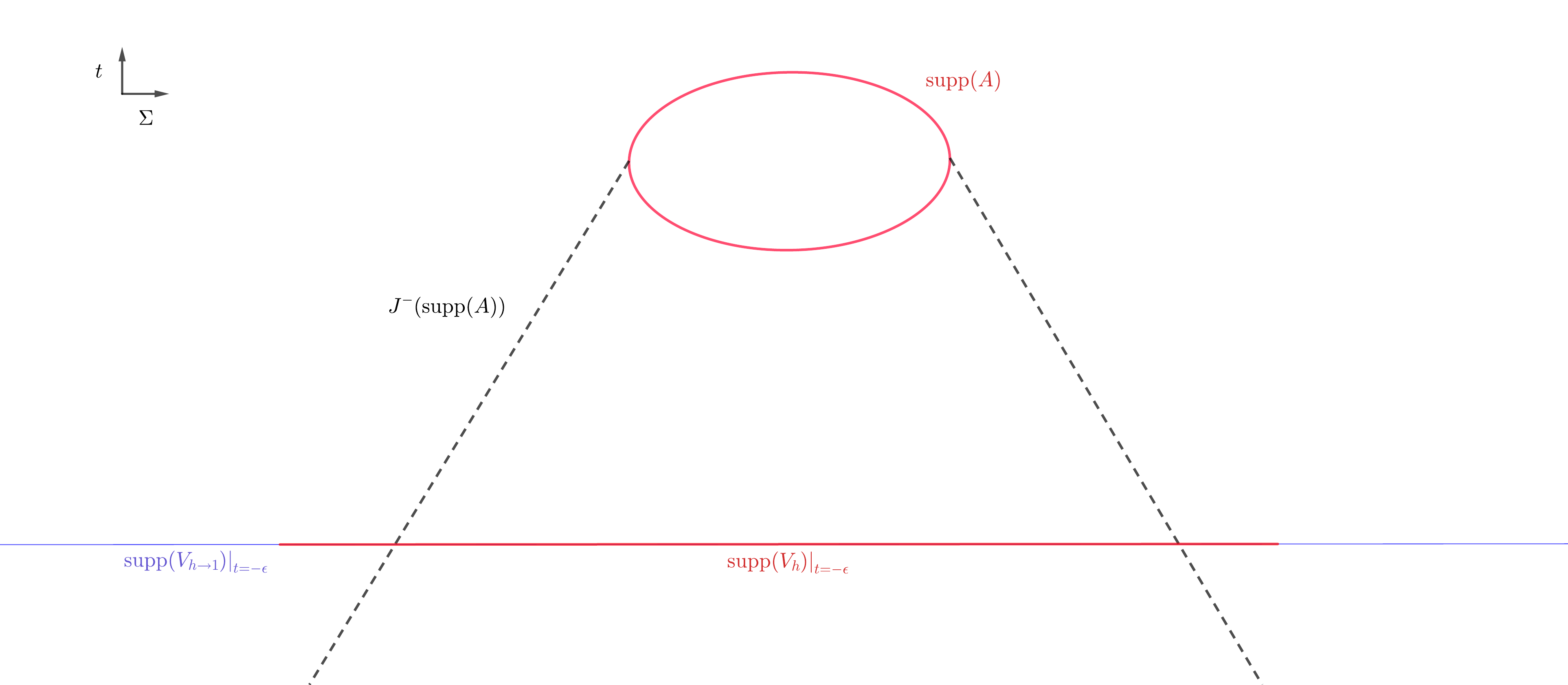}
 		\caption{The support of the interaction is in the future of the $t=-\epsilon$ drawn hypersurface of $\mathrm{supp}(V)$}
   \label{fig:1}
\end{figure}
Therefore, $R_{\Tilde{V}}(\Tilde{A})$ is always of compact support independently from $h$ provided $\Tilde{A} \in \mathfrak{A}^{\mathscr{G}_N}$.\\
In what follows we separate the discussion of KMS and ground states, by starting from the spacetime decay for the corresponding truncated $n$-point functions.

\subsection{KMS States}
\begin{prop}\label{thm: KMS}
Let $\omega^{\beta}$ be a KMS state with respect to $\tau_t$ on a realization of $\mathfrak{A}$, the algebra of smeared Wick polynomials, with translation invariant two point functions $\omega_2^{\beta,+}(x,y)$ and $\omega_2^{\beta,-}(x,y)$. Then, denoting by $\beta \mathcal{S}_n = \{(u_1, \ldots, u_n) \in \mathbb{R}^n : 0 \leq u_1 \leq \ldots \leq u_n \leq \beta \}$ for $0 < \beta < \infty$ and for all $A_0, \ldots, A_n \in \mathfrak{A}(\mathcal{O})$ with $\mathcal{O} \subset B_R \subset \mathbb{R}^4$, the truncated correlation functions:
\begin{equation}\label{eq: correlation}
    F_n^{\beta}(u_1, \mathbf{z}_1; \ldots; u_n, \mathbf{z}_n) = \omega^{\beta,\mathcal{T}}(A_0 \otimes \tau_{iu_1, \mathbf{z}_1} (A_1) \otimes \ldots \otimes \tau_{iu_n, \mathbf{z}_n} (A_n))
\end{equation}
\begin{itemize}
\item Decay exponentially in $\beta \mathcal{S}_n \times \mathbb{R}^{3n}$ if $m > 0$ and $r_e > 2R$:
\begin{equation*}
    |F_n^{\beta}(u_1, \mathbf{z}_1; \ldots; u_n, \mathbf{z}_n)| \leq c_{A_0, \ldots, A_n} e^{-\frac{m}{\sqrt{n}}r_e} \, , \quad r_e = \sqrt{\sum_{i=1}^n |\mathbf{z}_i^2|}.
\end{equation*}
for some $c_{A_0, \ldots, A_n} \in \mathbb{R}$ depending on the chosen observables.
\item Decay polynomially in $\beta \mathcal{S}_n \times \mathbb{R}^{3n}$ if $m = 0$:
\begin{equation*}
    |F_{n}^{\beta}(u_1, \mathbf{z}_1; \ldots; u_n, \mathbf{z}_n)| \leq c_{A_0, \ldots, A_n} \sum_{G \in \mathcal{G}_{n+1}^c} \prod_{l \in E(G)}\frac{1}{\big(1+|\mathbf{z}_{s(l)} - \mathbf{z}_{r(l)} |\big)^3},
\end{equation*}
for some $c_{A_0, \ldots, A_n} \in \mathbb{R}$ depending on the chosen observables. Here $G$ denotes a connected graph, with edges $l$ of source $s(l)$ and range $r(l)$, in the set of connected graphs with $n+1$ vertices.
\end{itemize}
\end{prop}
\begin{proof}
The proof is extensively discussed in the Appendix \ref{app: proof1}.
\end{proof}
In order to take the $h \to 1$ limit, by the discussion at the beginning of the section regarding the Bogoliubov map, we can study the adiabatic limit of the interacting state by replacing the generators $K$ given in \eqref{eq: generatore} with:
\begin{align*}
    K' &= - \lambda \int_{\mathbb{M}} \di^4x \, h(\mathbf{x}) \Dot{\chi}(t) R_{V_{h \to 1}} \big(\mathcal{L}_I\big)\\
    &= \lambda \int_{\mathbb{R}^3} \di^3 \mathbf{x} \, h(\mathbf{x}) \tau_{0,\mathbf{x}}(\mathcal{R}),
\end{align*}
where we have called for convenience:
\begin{equation*}
    \mathcal{R} \coloneqq - \int_{\mathbb{R}} \di t \, \Dot{\chi}(t) R_{V_{h \to 1}}(\mathcal{L}_I)
\end{equation*}
and we know this converges due to the compact support of $\dot{\chi}(t)$. From now on we denote by $\omega^{\beta,V'}$ the state with expansion \eqref{eq: espansionestato} where $K$ is replaced by $K'$. Moreover, we define a van Hove sequence of cutoff functions a sequence $(h_n)_{n \in \mathbb{N}}$ of test functions $h_n$ with the property:
\begin{equation*}
    0 \leq h_n(\mathbf{x}) \leq 1 \quad \forall \mathbf{x} \in \mathbb{R}^3 \, , \qquad h_n(\mathbf{x}) = \left\{\begin{aligned}
    &1 \quad |\mathbf{x}| < n\\
    &0 \quad |\mathbf{x}| > n + 1
    \end{aligned}\right. .
\end{equation*}
The following theorem proves the existence of the adiabatic limit for the constructed equilibrium states:
\begin{thm}\label{thm: 1}
Let $\omega^{\beta,V_h'}$ be the interacting KMS state as constructed in Theorem \ref{thm: stati} with respect to interaction of spatially compact support given by $h \in \mathcal{C}^{\infty}_{0}(\Sigma, \mathbb{R})$ and interacting dynamics generated by $K'$. Let $F_n^{\beta}$ be the analytically extended correlation functions as constructed in Equation \eqref{eq: correlation}. Then, for the massive and massless Dirac theory on Minkowski spacetime $F_{n}^{\beta}(u_1, \mathbf{z}_1, \ldots, u_n, \mathbf{z}_n) \in L^1(\beta \mathcal{S}_n \times \mathbb{R}^{3n})$. Moreover, given $\{h_k\}_{k \in \mathbb{N}} \in \mathcal{C}_0^{\infty}(\Sigma, \mathbb{R})$ an arbitrary sequence of van Hove cutoff functions, the adiabatic limit $k \to \infty$ exists at each order in perturbation theory:
\begin{equation*}
    \lim_{k \to \infty} \omega^{\beta, V_{h_k}'} = \omega^{\beta, V_{h \to 1}}.
\end{equation*}
\end{thm}
\begin{proof}
By the expansion in terms of connected correlation functions of the state $\omega^{\beta, V_{h_k}'}$, taking into account Remark \ref{rem: 2}, we have for any $A \in \mathfrak{A}$:
\begin{equation*}
    \omega^{\beta, V_{h_k}'}\big(R_{V_{h_k}}(A)\big) = \sum_{n=0}^{\infty} (-\lambda)^n \int_{\beta \mathcal{S}_n} \di^n u \int_{\mathbb{R}^{3n}} \di^{3n} \mathbf{x} \, h_k(\mathbf{x}_1) \ldots h_k(\mathbf{x}_n) \, \omega^{\beta, \mathcal{T}}\big( R_{V_{h_k}}(A) \otimes \tau_{iu_1, \mathbf{x}_1}(\mathcal{R}) \otimes \cdots \otimes \tau_{iu_n, \mathbf{x}_n}(\mathcal{R})\big).
\end{equation*}
The arguments of the above connected correlation function are formal power series in $\mathfrak{A}$, so just terms of the form $F_n^{\beta}$ appear. Therefore, as $k \to \infty$, the eventual integrability of $F_n^{\beta}$ ensures the existence of the adiabatic limit by dominated convergence theorem.\\
For the massive case, by the finite Lebesgue measure of $\beta \mathcal{S}_n$, we just have to focus on the integral over $\mathbb{R}^{3n}$. Let:
\begin{equation*}
    r_e = \sqrt{\sum_{i=1}^n |\mathbf{x}_i^2|},
\end{equation*}
and split the domain of integration in the two regions $r_e \leq 2 R$ and $r_e > 2R$. The integral over the first region is finite by the compactness of the region together with the smoothness of the integrand. For what concerns the second region of integration, convergence is ensured by the exponential decay of connected functions $F_n^{\beta}$ proven in Proposition \ref{thm: KMS}.\\
Also for the massless case, the finite Lebesgue measure of $\beta \mathcal{S}_n$ allows to focus just on the integrability on $\mathbb{R}^{3n}$. In this case, the result proven in Proposition \ref{thm: KMS} is uniform on $\mathbb{R}^{3n}$. Therefore, iteratively using Lemma \ref{lem: 2} taking into account that each graph $G$ is connected with a vertex at $(0,\mathbf{0})$, the integrability of $F_n^{\beta}$ is proven leading to the existence of the adiabatic limit.
\end{proof}

The just constructed equilibrium states in the adiabatic limit are independent from the choice of $\chi$ describing the way in which the interaction is switched on. A different choice $\chi'$ of the kind \eqref{eq:accad} with possibly also a different $\epsilon' > 0$ leads to the same equilibrium state $\omega^{\beta, V}$. The proof of this fact is not reported here as, with the above proven results, it will consist in the exact same steps as the one given in Proposition $5.3.$ of \cite{DragoHackPin}\\

\subsection{Ground State}
\begin{prop}\label{thm: ground}
Let $\omega^{\infty}$ be the ground state with respect to $\tau_t$ on $\mathfrak{A}$, the algebra of smeared Wick polynomials, with translation invariant two point functions $\omega^{\infty,+}_2(x,y)$ and $\omega^{\infty,-}_2(x,y)$. Then, denoting by $\mathcal{S}_n^{\infty} = \{ (u_1,\ldots, u_n) \in \mathbb{R}^n: -\infty < u_1 \leq \ldots \leq u_n \leq \infty \}$ and for all $A_0, \ldots, A_n \in \mathfrak{A}(\mathcal{O})$ with $\mathcal{O} \subset B_R \subset \mathbb{R}^4$, the connected correlation functions:
\begin{equation}\label{eq: connectedground}
    F_{n}^{\infty}(u_1, \mathbf{z}_1; \ldots; u_n, \mathbf{z}_n) = \omega^{\infty, \mathcal{T}}(\alpha_{iu_1, \mathbf{z}_1} (A_1) \otimes \ldots \otimes A_0 \otimes \ldots \otimes \alpha_{iu_n, \mathbf{z}_n} (A_n))
\end{equation}
\begin{itemize}
    \item Decay exponentially in $\mathcal{S}_n^{\infty} \times \mathbb{R}^{3n}$ if $m > 0$ and $r_g > 2R$:
    \begin{equation*}
        |F_n^{\infty}(u_1, \mathbf{z}_1; \ldots; u_n, \mathbf{z}_n)| \leq c_{A_0, \ldots, A_n} e^{-\frac{m}{\sqrt{n}}r_g} \, , \quad r_g = \sqrt{\sum_{i=1}^n u_i^2 + |\mathbf{z}_i^2|},
    \end{equation*}
    for some $c_{A_0, \ldots, A_n} \in \mathbb{R}$ depending on the chosen observables.
    \item Decay in $\mathcal{S}_n^{\infty} \times \mathbb{R}^{3n}$ if $m = 0$:
    \begin{equation*}
        |F_n^{\infty}(u_1, \mathbf{z}_1; \ldots; u_n, \mathbf{z}_n)| \leq c_{A_0, \ldots, A_n} \sum_{G \in \mathcal{G}^{c}_{n+1}}\prod_{l \in E(G)}\frac{1}{\bigg(1 + \sqrt{(u_{s(l)} - u_{r(l)})^2 + | \mathbf{z}_{s(l)} - \mathbf{z}_{r(l)} |^2}\bigg)^{3}}
    \end{equation*}
    for some $c_{A_0, \ldots, A_n} \in \mathbb{R}$ depending on the chosen observables. Here $G$ denotes a graph, with edges $l$ of source $s(l)$ and range $r(l)$ in the set of connected graphs with $n+1$ vertices.
\end{itemize}
\end{prop}
\begin{proof}
The proof is extensively discussed in Appendix \ref{app: proof2}
\end{proof}
As we did for the case at finite inverse temperature, adopting the same notation, using the expansion \eqref{eq: espansionestato} for the interacting state we can take the limit $h \to 1$ removing the assumption on the spatial localization. 
\begin{thm}\label{thm: 2}
Let $\omega^{\infty,V_h'}$ be the interacting ground state as constructed in Theorem \ref{thm: stati} with respect to interaction of spatially compact support given by $h \in \mathcal{C}^{\infty}_{0}(\Sigma, \mathbb{R})$ and interacting dynamics generated by $K'$. Let $F_n^{\infty}$ be the analytically extended correlation functions as constructed in Equation \eqref{eq: connectedground}. Then, for the massive Dirac theory on Minkowski spacetime $F_{n}^{\infty}(u_1, \mathbf{z}_1, \ldots, u_n, \mathbf{z}_n) \in L^1(\mathcal{S}_n^{\infty} \times \mathbb{R}^{3n})$. Moreover, given $\{h_k\}_{k \in \mathbb{N}} \in \mathcal{C}_0^{\infty}(\Sigma, \mathbb{R})$ an arbitrary sequence of van Hove cutoff functions, the adiabatic limit $k \to \infty$ exists at each order in perturbation theory:
\begin{equation*}
    \lim_{k \to \infty} \omega^{\infty, V_{h_k}'} = \omega^{\infty, V_{h \to 1}}.
\end{equation*}
\end{thm}
\begin{proof}
By the expansion in connected correlation functions of the state $\omega^{\infty, V_{h_k}'}$ we know that for any $A \in \mathfrak{A}$:
\begin{align}
    &\omega^{\infty, V_{h_k}'}\big(R_{V_{h_k}}(A)\big) \nonumber\\ 
    &=\sum_{n=0}^{\infty} (-\lambda)^n \int_{\mathcal{S}_{n,\infty}} \di^n u \int_{\mathbb{R}^{3n}} \di^{3n}\mathbf{x} \, h_k(\mathbf{x}_1) \cdots h_k(\mathbf{x}_n) \, \omega^{\infty,\mathcal{T}} \bigg( \bigotimes_{j: u_j < 0} \tau_{iu_j, \mathbf{x}_j}(\mathcal{R}) \otimes R_{V_{h_k}}(A) \otimes \bigotimes_{j: u_j \geq 0}\tau_{iu_j, \mathbf{x}_j}(\mathcal{R}) \bigg), \label{eq: groundAd}
\end{align}
where $\mathcal{S}_{n, \infty} \coloneqq \{ (u_1, \ldots, u_n) \in \mathbb{R}^n : -\infty \leq u_1 \leq \ldots \leq u_n \leq \infty \}$. Being the arguments of the truncated correlation functions formal power series in $\mathfrak{A}$, at each fixed perturbative order the integrand is made out of terms of the form $F_n^{\infty}$. Therefore, if such functions are proven to be absolutely integrable the limit $k \to \infty$ exists finite. First of all, by positivity, the integration region is extended to $\mathbb{R}^{4n}$.\\
For the massive case, call:
\begin{equation*}
    r_g = \sqrt{\sum_{i = 1}^{n} u_i^2 + |\mathbf{x}_i|^2}.
\end{equation*}
Split the domain of integration in the regions $r_g \leq 2R$ and $r_g > 2R$. The integral over the first region is finite by the compactness of the region together with the smoothness of the integrand. For the second region of integration, the exponential decay in Proposition \ref{thm: ground} ensures the convergence.
\end{proof}
\begin{rem}
For the the massless Dirac theory on Minkowski spacetime the same result holds, provided we restrict to a certain, quite general, subclass of interacting theories. Namely, given $P: \mathbb{R}^3 \to \mathbb{R}$ so that $V$ in Equation \eqref{eq: formaV} is a microcausal functional, the adiabatic limit exists for interacting models with $\mathcal{L}_I = (\overline{\psi} \cancel{P}(\mathbf{x}) \psi)^n$ if $n > 1$ or, if $n=1$ provided that $|P(\mathbf{x})| \leq C(1 + |\mathbf{x}|)^{-1}$ for some $C \in \mathbb{R}^+$.\\
The proof is equivalent to the above using the decay for the massless case shown in Proposition \ref{thm: ground} and the argument of Lemma \ref{lem: 2}. In particular, for $n>1$ the number of lines between different vertices of the diagrams, ensures that $F_{n}^{\infty}(u_1, \mathbf{z}_1; \ldots; u_n, \mathbf{z}_n) \in L^1(\mathbb{R}^{4n})$ while, for $n=1$, the assumption on $P(\mathbf{x})$ allows to reduce to the same argument as in the proof of Theorem \ref{thm: 1}.
\end{rem}
The state $\omega^{\infty, V_{h \to 1}}$ is again independent from the choice of $\chi$, as it is constructed as the $\beta \to \infty$ limit of a state that is independent from this choice.

\section{Linear Response and Debye Scattering Length}\label{sec: DebyeLeng}
The above constructed states are applied in this section to the study of the expectation value of the conserved current of the Dirac field on a thermal equilibrium state and to solutions of the semiclassical Maxwell equations with a quantized source. The interacting theory that is studied is given by the QED-coupling:
\begin{equation}\label{eq: couplQED}
    V = \lambda e \int_{\mathbb{M}} \di^4x \, h(\mathbf{x}) \chi(t) A_{\mu}(x) \overline{\psi}(x) \gamma^{\mu} \psi(x)
\end{equation}
where $A_{\mu}$ is a classical electromagnetic potential such that for any fixed $\mu = 0,1,2,3$ has finite uniform norm $\| A_{\mu} \|_{\infty} < +\infty$ and is such that $WF(\dot{\chi} A_{\mu}) \subset \overline{V}^+ \cup \overline{V}^-$ for any component $\mu$. Moreover, $\chi \in \mathcal{C}^{\infty}(\mathbb{R}, \mathbb{R})$ is the switch-on function as in \eqref{eq:accad}, and $h \in \mathcal{C}^{\infty}_0(\Sigma, \mathbb{R})$ the space cutoff. We will further assume that the electromagnetic potential is stationary $A_{\mu}(x) = A_{\mu}(\mathbf{x})$. The study of charged Fermi fields in the presence of external electromagnetic potential is reviewed in \cite{IeidesGrotch,FedotovIlderton}. In the framework of AQFT charged Fermi fields in external electromagnetic fields have been been studied in various works \cite{AbramBrunetti, SchlemmerZahn, FrobZahn}.\\

We aim at finding solutions, for the gauge potential $A^{\mu}(\mathbf{x})$, of the following system of semiclassical PDEs:
\begin{equation}\label{eq: Semiclass}
     \Delta A^\mu(\mathbf{x}) = -\left\langle j^{\mu}_{\mathrm{q}}(\mathbf{x}) \right\rangle_{\beta,V} - j^{\mu}_{\mathrm{class}}(\mathbf{x}).
\end{equation}
Here, $j^{\mu}_{\mathrm{class}} \in \mathcal{C}_0^{\infty}(\mathbb{M}, T\mathbb{M})$ is the classical source of the existing classical background electromagnetic field, while $\left\langle j^{\mu}_{\mathrm{q}} \right\rangle_{\beta,V}$ is the quantum Dirac field source, at thermal equilibrium and in interaction with the background field via \eqref{eq: couplQED}, defined as:
\begin{equation*}
    \left\langle j^{\mu}_{\mathrm{q}}(\mathbf{x}) \right\rangle_{\beta,V} = \omega^{\beta, V}\big( R_V(j^{\mu}(\mathbf{x}))\big).
\end{equation*}
Finally, we have implicitly assumed that the gauge fields satisfy the Lorenz gauge condition: $\partial_{\mu} A^{\mu} = 0$. In particular, we remark that by the above construction of $\omega^{\beta,V}$, the solutions of \eqref{eq: Semiclass} are independent from the switching $\chi$ and can be determined in the limit $h \to 1$. Moreover, by the perturbative definition of $\omega^{\beta,V}$, the right hand side of the above equation can be computed at the desired perturbative order in the coupling $\lambda$. Here, and in the following, for the sake of simplicity and for its interest in linear response theory \cite{LeBellac, KaputsaGale} we will consider just the first order.\\\\

The starting point is the expansion at first order of both the Bogoliubov map and the unitary cocycle $U_V(i \beta)$. The explicit computation is omitted here, as the strategy is essentially the same as in the proof of Proposition $4.7$ of \cite{GalandaSangalettiPin} to which we refer the interested reader. The results are for $\mu = 0$:
\begin{equation*}
    \omega^{\beta,V,(1)}\big(R_V(j^0(\mathbf{x}))\big) = I^0(\mathbf{x}) + J^0(\mathbf{x}),
\end{equation*}
and for $\mu = k \in \{ 1,2,3 \}$:
\begin{equation*}
     \omega^{\beta,V,(1)}\left( R_V\big(j^k(\mathbf{x})\big) \right) = I^k(\mathbf{x}) + J^k(\mathbf{x}),
\end{equation*}
where:
\begin{align*}
    I^0(\mathbf{x}) &= \frac{4\lambda e^2}{(2 \pi)^6} \int_{\mathbb{R}^3 \times \mathbb{R}^3} \frac{\di^3 \mathbf{p} \di^3 \mathbf{k}}{\en{p}^2 - \en{k}^2} \hat{A}^{0}(\mathbf{k} - \mathbf{p}) e^{-i(\mathbf{p} - \mathbf{k})\mathbf{x}} \left[\frac{\en{p}^2 + m^2 + k_i p^i}{\en{p} (1 + e^{\beta \en{p}})} - \frac{\en{k}^2 + m^2 + k_i p^i}{\en{k} (1 + e^{\beta \en{k}})}\right]\\
    J^0(\mathbf{x}) &= \lambda e^2 a_1 \Delta A^0(\mathbf{x})\\
    I^k(\mathbf{x}) &= -\frac{4\lambda e^2}{(2 \pi)^6} \int_{\mathbb{R}^3 \times \mathbb{R}^3} \frac{\di^3 \mathbf{p} \di^3 \mathbf{k}}{\en{p}^2 - \en{k}^2} \hat{A}^{k}(\mathbf{k} - \mathbf{p}) e^{-i(\mathbf{p} - \mathbf{k})\mathbf{x}} \left[\frac{\en{p}^2 - m^2 - k_i p^i}{\en{p} (1 + e^{\beta \en{p}})} - \frac{\en{k}^2 - m^2 - k_i p^i}{\en{k} (1 + e^{\beta \en{k}})}\right]\\
    J^k(\mathbf{x}) &= \frac{16 \lambda e^2}{3(2 \pi)^5} \left(\Box \Box \int_{(2m)^2}^{\infty} \frac{\di M^2}{M^5} \left( \frac{M^2}{4} - m^2\right)^{\frac{1}{2}}\left( \frac{M^2}{4} + \frac{m^2}{2}\right) \int_{\mathbb{R}^3} \frac{\di^3 \mathbf{k}}{\omega_{\mathbf{k},M}^2} \hat{A}^k(-\mathbf{k}) e^{-i\mathbf{k}\mathbf{x}} \right) + \lambda e^2 a_1 \Delta A^k(\mathbf{x}) 
\end{align*}
Here, the adiabatic limit, $h \to 1$, has already been taken as a consequence of the proved convergence at each perturbative order (Theorem \ref{thm: 1}) and $\hat{A}^{\mu}(\mathbf{k})$ is the Fourier transform of the electromagnetic potential $A^{\mu}(\mathbf{x})$, if necessary considered in the weak sense. Moreover, $\omega_{\mathbf{k},M} = \sqrt{|\mathbf{k}|^2 + M^2}$ and $a_1 \in \mathbb{R}$ is the renormalization constant arising from the ambiguities in the extension of the products of Feynman propagators or equivalently, using the perturbative agreement for quadratic interacting Lagrangian \cite{HollandsWald2001, DragoHackPin}, by the point splitting regularization in the definition of the current. This is the only renormalization freedom left after imposing the conservation of the current (see \cite{Heisenberg, SchlemmerZahn}).\\
The $I^{\mu}(\mathbf{x})$ terms encode all the thermal effects as they vanish for $\beta \to \infty$, while $J^{\mu}(\mathbf{x})$ are the zero temperature contributions. Finally, for $\mu = 0$, the zero temperature contributions involve only the renormalization freedoms as a consequence of the Ward identities combined with the K\"allen-Lehmann procedure that was used for the extension of the distributions.\\

The general solution of \eqref{eq: Semiclass} can be expressed via the convolution theorem. In particular, taking the Fourier transform we have:
\begin{equation*}
    \hat{I}^{\mu}(\Tilde{\mathbf{p}}) = \lambda \hat{A}^{\mu}(\Tilde{\mathbf{p}})\hat{\mathscr{F}}_{\mu}(\Tilde{\mathbf{p}}) \,,  \quad \hat{J}^{\mu}(\Tilde{\mathbf{p}}) = \lambda \hat{A}^{\mu}(\Tilde{\mathbf{p}}) \hat{\mathscr{B}}_{\mu}(\Tilde{\mathbf{p}}),
\end{equation*}
where $\Tilde{\mathbf{p}}$ is understood as the external momentum and:
\begin{equation*}
    \hat{\mathscr{F}}_{\mu}(\Tilde{\mathbf{p}}) \coloneqq \left\{\begin{aligned}
        \frac{4 e^2}{(2 \pi)^3}\int_{\mathbb{R}^3} \frac{\di^3 \mathbf{p} }{\en{p}^2 - \omega_{\Tilde{\mathbf{p}} + \mathbf{p}}^2} \left[\frac{\en{p}^2 + m^2 + (\Tilde{p} + p)_i p^i}{\en{p} (1 + e^{\beta \en{p}})} - \frac{\omega_{\Tilde{\mathbf{p}} + \mathbf{p}}^2 + m^2 + (\Tilde{p} + p)_i p^i}{\omega_{\Tilde{\mathbf{p}} + \mathbf{p}}^2 (1 + e^{\beta \omega_{\Tilde{\mathbf{p}} + \mathbf{p}}})}\right] \hspace{30pt} \mathrm{for} \,\, \mu = 0 \\
        - \frac{4 e^2}{(2 \pi)^3}\int_{\mathbb{R}^3} \frac{\di^3 \mathbf{p} }{\en{p}^2 - \omega_{\Tilde{\mathbf{p}} + \mathbf{p}}^2} \left[\frac{\en{p}^2 - m^2 - (\Tilde{p} + p)_i p^i}{\en{p} (1 + e^{\beta \en{p}})} - \frac{\omega_{\Tilde{\mathbf{p}} + \mathbf{p}}^2 - m^2 - (\Tilde{p} + p)_i p^i}{\omega_{\Tilde{\mathbf{p}} + \mathbf{p}}^2 (1 + e^{\beta \omega_{\Tilde{\mathbf{p}} + \mathbf{p}}})}\right] \hspace{30pt} \mathrm{for} \,\, \mu = k
    \end{aligned}\right..
\end{equation*}
Moreover, $\Box$ is nothing but the Laplacian $\partial_i \partial^i$ because it acts on distributions over $\mathcal{D}(\mathbb{R}^3)$ constant in time. So, by Plancherel theorem:
\begin{align}
    \hat{\mathscr{B}}_{\mu}(\Tilde{\mathbf{p}}) \coloneqq  e^2 a_1 |\Tilde{\mathbf{p}}|^2 
    +\left\{\begin{aligned}
    &0 \hspace{30pt} \mathrm{for} \,\, \mu = 0\\
    &\frac{16 e^2 |\Tilde{\mathbf{p}}|^4}{3(2 \pi)^5} \int_{(2m)^2}^{\infty} \frac{\di M^2}{M^5 \omega_{\Tilde{\mathbf{p}},M}^2} \left( \frac{M^2}{4} - m^2\right)^{\frac{1}{2}}\left( \frac{M^2}{4} + \frac{m^2}{2} \right) \hspace{30pt} \mathrm{for} \,\, \mu = k
    \end{aligned} \right. \label{eq: parteRin}.
\end{align}
Finally, denoting by $\mathcal{F}^{-1}$ the inverse Fourier transform and by $\ast$ the convolution, we get the general solution of \eqref{eq: Semiclass}:
\begin{equation}\label{eq: solSemiclass}
    A^{\mu}(\mathbf{x}) = \left(j^{\mu}_{\mathrm{class}} \ast \mathcal{F}^{-1}\left[ \frac{1}{|\Tilde{\mathbf{p}}|^2 - \lambda \left( \hat{\mathscr{F}}_{\mu}(\Tilde{\mathbf{p}}) + \hat{\mathscr{B}}_{\mu}(\Tilde{\mathbf{p}})\right)} \right] \right)(\mathbf{x})
\end{equation}

To get an explicit correction of the electromagnetic potential due to this backreaction, we expand the solution in $|\Tilde{\mathbf{p}}|$. At zeroth order, the term \eqref{eq: parteRin} vanishes:
\begin{equation*}
    \hat{\mathscr{B}}_{\mu}(\Tilde{\mathbf{p}})\big|_{|\Tilde{\mathbf{p}}| = 0} = 0 \quad \forall \mu = 0,1,2,3
\end{equation*}
While for what concerns $ \hat{\mathscr{F}}_{\mu}(\Tilde{\mathbf{p}})$, expanding in $|\Tilde{\mathbf{p}}|$ at zeroth order, we first notice that:
\begin{equation*}
    \hat{\mathscr{F}}_k(\Tilde{\mathbf{p}})|_{|\Tilde{\mathbf{p}}| = 0} = 0 \quad \mathrm{for} \,\,\,  k = 1,2,3.
\end{equation*}
While, calling $-\hat{\mathscr{F}}_0(\Tilde{\mathbf{p}})|_{|\Tilde{\mathbf{p}}| = 0} = m_D^2$, we have:
\begin{align*}
    m_D^2 &= -\frac{4 e^2 }{(2 \pi)^3} \int_{\mathbb{R}^3} \di^3 \mathbf{p} \, \partial_{\omega_{\mathbf{p}}}\bigg( \frac{1}{1 + e^{\beta \omega_{\mathbf{p}}}} \bigg)\\
    &= \frac{4e^2 m^2}{(2 \pi)^2} \left(2 \sum_{n=0}^{\infty} (-1)^n \int_1^{\infty} \di x \, \sqrt{x^2 - 1} \, e^{-(n+1) \beta m x} + \sum_{n=0}^{\infty} (-1)^n \int_{1}^{\infty} \di x \, \frac{e^{-(n+1) \beta m x}}{\sqrt{x^2 - 1}} \right)\\
    &= \frac{4e^2 m^2 }{(2 \pi)^2} \sum_{n=0}^{\infty} (-1)^{n} K_2\big( (n+1) \beta m\big) \numberthis \label{eq: massafotoni}
\end{align*}
where we used \eqref{eq: FermiFactor} to write the Fermi factor and introduced the modified Bessel functions of second kind and index $n \in \mathbb{N}$. At the last step we used the properties for Bessel function, see \cite{gradshteyn2007} Equation $(17)$.\\
We do the following remarks. For $\beta \to \infty$, expression \eqref{eq: massafotoni} vanishes confirming the thermal nature of the effect. For the $m \to 0$ limit, consider the series expansion of the modified Bessel functions of second kind and index $n \in \mathbb{N}$ for small argument (see Equation $8.446$ of \cite{gradshteyn2007}):
\begin{equation*}
    K_n(z) = \frac{1}{2}\sum_{k=0}^{n-1} (-1)^k \frac{(n-k-1)!}{k! \left(\frac{z}{2}\right)^{n-2k}} + (-1)^{n+1} \sum_{k=0}^{\infty}\frac{\left( \frac{z}{2} \right)^{n+2k}}{k! (n+k)!} \bigg[ \ln \frac{z}{2} - C(n,k)\bigg],
\end{equation*}
where $C(n,k)$ is a factor depending on the naturals $n$ and $k$.
Therefore:
\begin{equation*}
    \lim_{m \to 0} m_D^2 = \frac{e^2}{6 \beta^2}. 
\end{equation*}
Clearly, the same conclusion is reached in the high temperature ($\beta \to 0$) limit.\\\\
The quantity $m_D^2$, is the \textit{Debye mass} \cite{DebyeHuckel} when the electromagnetic field is coupled with a Dirac field in the case of arbitrary mass and temperature. In fact, the above massless and infinite temperature limits coincide with the results known in literature (e.g. Equation $6.44$ of \cite{LeBellac} and Chapter 6 of \cite{KaputsaGale}). We also point out that, in the vanishing external momentum approximation by the form of the solution \eqref{eq: solSemiclass}, only the electric field is corrected. In support of the interpretation of our result as the Debye mass, we notice that the solution of the semiclassical Maxwell equation \eqref{eq: solSemiclass} becomes a screening of the background potential. Namely, for $\mu = 0$ the inverse Fourier transform coincides with a Yukawa potential:
\begin{equation*}
    \mathcal{F}^{-1}\bigg[ \frac{1}{|\Tilde{\mathbf{p}}|^2 + \lambda m_D^2} \bigg] = \frac{1}{4\pi r}e^{-\sqrt{\lambda} r/\mathscr{\lambda}_D},
\end{equation*}
where we have defined:
\begin{equation*}
    \lambda_D \coloneqq \frac{1}{m_D}
\end{equation*}
the \textit{Debye screening length}.\\

As an example, the specific case of the electromagnetic field generated by a point-like classical source can be studied considering as classical source:
\begin{equation*}
    j^0_{\mathrm{class}}(\mathbf{x}) = q \delta_{\epsilon}(\mathbf{x})
\end{equation*}
where $q$ is the total charge and $\delta_{\epsilon}$ a class of compactly supported smooth functions that in the limit $\epsilon \to 0$ converges weakly to the Dirac delta distribution:
\begin{equation*}
    \lim_{\epsilon \to 0} \delta_{\epsilon}(\mathbf{x}) = \delta(\mathbf{x}).
\end{equation*}
Therefore, taking $\epsilon \to 0$ in \eqref{eq: solSemiclass}, by continuity of the convolution, the general solution becomes:
\begin{equation*}
    \lim_{\epsilon \to 0} q \left( \delta_{\epsilon} \ast  \mathcal{F}^{-1}\bigg[ \frac{1}{|\Tilde{\mathbf{p}}|^2 + \lambda m_D^2} \bigg] \right)(\mathbf{x}) = \frac{q}{4\pi r}e^{-\sqrt{\lambda} r/\mathscr{\lambda}_D}
\end{equation*}
Namely, the effect of the backreaction to the Coulomb potential results in a Yukawa potential and correspondingly in the screening of the electric field.

\appendix
\section{Proof of Propositions \ref{thm: KMS} and \ref{thm: ground}}
In order to ease the reading of the proofs, we repeat here the statement of the propositions and present their proofs.
\begin{prop}
Let $\omega^{\beta}$ be a KMS state with respect to $\tau_t$ on a realization of $\mathfrak{A}$, the algebra of smeared Wick polynomials, with translation invariant two point functions $\omega_2^{\beta,+}(x,y)$ and $\omega_2^{\beta,-}(x,y)$. Then, denoting by $\beta \mathcal{S}_n = \{(u_1, \ldots, u_n) \in \mathbb{R}^n : 0 \leq u_1 \leq \ldots \leq u_n \leq \beta \}$ for $0 < \beta < \infty$ and for all $A_0, \ldots, A_n \in \mathfrak{A}(\mathcal{O})$ with $\mathcal{O} \subset B_R \subset \mathbb{R}^4$, the truncated correlation functions:
\begin{equation}
    F_n^{\beta}(u_1, \mathbf{z}_1; \ldots; u_n, \mathbf{z}_n) = \omega^{\beta,\mathcal{T}}(A_0 \otimes \tau_{iu_1, \mathbf{z}_1} (A_1) \otimes \ldots \otimes \tau_{iu_n, \mathbf{z}_n} (A_n))
\end{equation}
\begin{itemize}
\item Decay exponentially in $\beta \mathcal{S}_n \times \mathbb{R}^{3n}$ if $m > 0$ and $r_e > 2R$:
\begin{equation*}
    |F_n^{\beta}(u_1, \mathbf{z}_1; \ldots; u_n, \mathbf{z}_n)| \leq c_{A_0, \ldots, A_n} e^{-\frac{m}{\sqrt{n}}r_e} \, , \quad r_e = \sqrt{\sum_{i=1}^n |\mathbf{z}_i^2|}.
\end{equation*}
for some $c_{A_0, \ldots, A_n} \in \mathbb{R}$ depending on the chosen observables.
\item Decay polynomially in $\beta \mathcal{S}_n \times \mathbb{R}^{3n}$ if $m = 0$:
\begin{equation*}
    |F_{n}^{\beta}(u_1, \mathbf{z}_1; \ldots; u_n, \mathbf{z}_n)| \leq c_{A_0, \ldots, A_n} \sum_{G \in \mathcal{G}_{n+1}^c} \prod_{l \in E(G)}\frac{1}{\big(1+|\mathbf{z}_{s(l)} - \mathbf{z}_{r(l)} |\big)^3},
\end{equation*}
for some $c_{A_0, \ldots, A_n} \in \mathbb{R}$ depending on the chosen observables. Here $G$ denotes a connected graph, with edges $l$ of source $s(l)$ and range $r(l)$, in the set of connected graphs with $n+1$ vertices.
\end{itemize}
\end{prop}
\begin{proof}\label{app: proof1}
The connected correlation functions can be rewritten as a sum on the connected graphs $G$ in the set of all connected graphs with $n+1$ vertices $\mathcal{G}_{n+1}^c$:
\begin{equation*}
    F_n^{\beta}(u_1, \mathbf{z}_1; \ldots; u_n, \mathbf{z}_n) = \sum_{G \in \mathcal{G}_{n+1}^c} \frac{1}{\mathrm{Symm}(G)} F^{\beta}_{n,G}(u_1, \mathbf{z}_1; \ldots; u_n, \mathbf{z}_n).
\end{equation*}
In the above, $\mathrm{Symm}(G)$ is the symmetry factor associated to each graph and:
\begin{equation}\label{eq: AA}
    F^{\beta}_{n,G}(u_1, \mathbf{z}_1; \ldots; u_n, \mathbf{z}_n) = \prod_{i < j} \big( \Gamma_{\beta,2}^{ij} \big)^{l_{ij}} (A_0 \otimes \tau_{iu_1, \mathbf{z}_1}(A_1) \otimes  \cdots \otimes \tau_{iu_n, \mathbf{z}_n}(A_n))\bigg|_{\mathrm{Conf} = 0}
\end{equation}
where $l_{ij}$ is the number of lines in the graph connecting the $i$-th and $j$-th vertex, $\mathrm{Conf}$ is a shortening for all the field configurations set to zero by the evaluation on the state $\omega^{\beta}$, and $\Gamma_{\beta,2}^{ij}$ is a functional differential operator coming from the realization of the $\star_{\omega^{\beta}}$-product:
\begin{equation}\label{eq: Gamma}
    \Gamma_{\beta,2}^{ij} = \hbar \int \di^4x \di^4y \bigg( \omega^{\beta,+}_2(x-y) \frac{\delta_r}{\delta \psi(x)} \otimes \frac{\delta}{\delta \overline{\psi}(y)} + \omega^{\beta,-}_2(y-x) \frac{\delta_r}{\delta\overline{\psi}(x)} \otimes \frac{\delta}{\delta{\psi}(y)}\bigg).
\end{equation}
The explicit form of the two point functions is:
\begin{align*}
    \omega^{\beta, +}_{2}(x-y) &= \frac{1}{(2 \pi)^3} \int \frac{\di^3\mathbf{p}}{2 \omega_p} \bigg( \frac{(-\gamma^0 \omega_p - \gamma^i p_i + m)e^{-i\omega_p (t_x - t_y)}}{1 + e^{-\beta \omega_p}} - \frac{(\gamma^0 \omega_p - \gamma^i p_i + m)e^{i\omega_p (t_x - t_y)}}{1 + e^{\beta \omega_p}}  \bigg) e^{i \mathbf{p} (\mathbf{x}- \mathbf{y})}\\
    \omega^{\beta, -}_{2}(x-y) &= \frac{1}{(2 \pi)^3} \int \frac{\di^3\mathbf{p}}{2 \omega_p} \bigg(\frac{(-\gamma^0 \omega_p - \gamma^i p_i + m)e^{-i\omega_p (t_x - t_y)}}{1 + e^{\beta \omega_p}} - \frac{(\gamma^0 \omega_p - \gamma^i p_i + m)e^{i\omega_p (t_x - t_y)}}{1 + e^{-\beta \omega_p}}  \bigg) e^{i \mathbf{p} (\mathbf{x}- \mathbf{y})}
\end{align*}
where $\omega_p = \sqrt{|\mathbf{p}|^2 + m^2}$.\\
Now, instead of making the product over vertices, we write \eqref{eq: AA} as a product over lines. Introduce for $N = |E(G)|$ the number of lines in the graph $G$, the multiindex $\mathbf{K} = (k_1, \ldots, k_N)$ where each $k_s = \pm$. Therefore, summing over all possible multiindices $\mathbf{K}$, the above is rewritten as:
\begin{equation*}
    F^{\beta}_{n,G}(u_1, \mathbf{z}_1; \ldots; u_n, \mathbf{z}_n) = \sum_{\mathbf{K}} \int \di X \di Y \prod_{l \in E(G)} \omega_2^{\beta, k_{l}}(x_l,y_l) \Psi^{\mathbf{K}}(X,Y)\bigg|_{\mathrm{Conf} = 0}.
\end{equation*}
Here, $\di X = \di^4x_1 \cdots \di^4x_N$, $\omega_2^{\beta, k_{l}}(x_l,y_l)$ denotes the two point function with difference of the arguments depending on $k_l$ and:
\begin{equation*}
    \Psi^{\mathbf{K}}(X,Y) \coloneqq \prod_{l \in E(G)} \mathfrak{D}^{k_l}(x_l, y_l)(A_0 \otimes \alpha_{iu_1, \mathbf{z}_1} (A_1) \otimes \ldots \otimes \alpha_{iu_n, \mathbf{z}_n} (A_n)),
\end{equation*}
where $\mathfrak{D}^{k_l}(x_l,y_l)$ are the functional derivatives according to $k_l$ as expressed in Equation \eqref{eq: Gamma}. Changing now integration variables, moving the spacetime translations to the arguments of the two point functions, writing everything in momentum space and using the explicit form of the two point functions:
\begin{equation}\label{eq: Mos}
    F^{\beta}_{n,G}(U,\mathbf{Z}) = \sum_{\mathbf{K}} \int \di P \prod_{l \in E(G)} e^{i\mathbf{p}_l (\mathbf{z}_{s(l)} - \mathbf{z}_{r(l)})} \bigg(\frac{\lambda_+^{k_l}(p_l) - \lambda_-^{k_l}(p_l)}{2 \omega_{p_l} (1 + e^{-\beta \omega_{p_l}})}\bigg) \hat{\Psi}(-P,P),
\end{equation}
where $s(l) < r(l)$ are respectively the source and the range vertices of the line $l$ and:
\begin{align*}
    \lambda_+^{+}(p_l) &= e^{\omega_{p_l} (u_{s(l)} - u_{r(l)})} \delta(p_l^0 - \omega_{p_l})(-\cancel{p}_l + m)\\
    \lambda_-^{+}(p_l) &= e^{-\beta \omega_{p_l}} e^{\omega_{p_l} (u_{r(l)} - u_{s(l)})} \delta(p_l^0 + \omega_{p_l})(-\cancel{p}_l + m),
\end{align*}
while:
\begin{align*}
    \lambda_+^{-}(p_l) &= -e^{\omega_{p_l} (u_{s(l)} - u_{r(l)})} \delta(p_l^0 + \omega_{p_l})(-\cancel{p}_l + m)\\
    \lambda_-^{-}(p_l) &= -e^{-\beta \omega_{p_l}} e^{\omega_{p_l} (u_{r(l)} - u_{s(l)})} \delta(p_l^0 - \omega_{p_l})(-\cancel{p}_l + m).
\end{align*}
Notice that the order of the difference in the exponential is the same in both cases due to the corresponding switched order of the difference $x_l - y_l$ in the two-point functions.\\
Now, by the KMS property of the free connected two-point function:
\begin{equation*}
    \omega^{\beta,\mathcal{T}}(\alpha_{0, \mathbf{0}} A_0 \otimes \ldots \otimes \alpha_{iu_n, \mathbf{z}_n}A_n) = \omega^{\beta,\mathcal{T}}(\alpha_{i u_m, \mathbf{z}_m} A_m \otimes \ldots \otimes \alpha_{iu_n, \mathbf{z}_n}A_n \otimes \alpha_{i\beta}A_0 \otimes \ldots \otimes \alpha_{i(u_{m-1} + \beta), \mathbf{z}_{m-1}} A_{m-1}),
\end{equation*}
where the choice of $m \in \{0, \ldots, n\}$ is arbitrary. Therefore, recalling that $(u_0, \ldots, u_n) \in \beta \mathcal{S}_{n+1}$ one of the following conditions hold:
\begin{equation}\label{eq: u-conditions}
    \beta - u_n \geq \frac{\beta}{n+1}\, ; \quad u_m - u_{m-1} \geq \frac{\beta}{n+1},
\end{equation}
where the second condition holds for at least an $m \in \{0, \ldots, n\}$. However, the two are equivalent. In fact, if the second condition holds, redefine:
\begin{align*}
    (\mathbf{w}_1, \ldots, \mathbf{w}_n) &= (\mathbf{z}_{m+1} - \mathbf{z}_m, \ldots, \mathbf{z}_n - \mathbf{z}_m, - \mathbf{z}_m, \mathbf{z}_1 - \mathbf{z}_m, \ldots, \mathbf{z}_{m-1} - \mathbf{z}_m)\\
    (v_1, \ldots, v_n) &= (u_{m+1} - u_m, \ldots, u_{n} - u_m, \beta - u_m, \beta + u_1 - u_m, \ldots, \beta + u_{m-1} - u_m)\\
    B_0 &\coloneqq A_m\,, \,\,\, B_1 \coloneqq A_{m+1}\, , \,\,\, \cdots\,\,\, , \,\,\, B_n \coloneqq A_{m-1}.
\end{align*}
Then, from the spacetime translation invariance of the KMS state, we rewrite:
\begin{align*}
    F^{\beta}_{n,G}(U, \mathbf{Z}) = \overline{F}^{\beta}_{n,G}(V, \mathbf{W}) &\coloneqq \omega^{\beta,\mathcal{T}}(B_0 \otimes \alpha_{iv_1, \mathbf{w}_1} B_1 \otimes \ldots \otimes \alpha_{iv_n, \mathbf{w}_n} B_n)\\
    &= \sum_{\mathbf{K}} \int \di P \prod_{l \in E(G)} e^{ik_l \mathbf{p}_l (\mathbf{w}_{s(l)} - \mathbf{w}_{r(l)})} \frac{\big(\overline{\lambda}^{k_l}_+(p_l) - \overline{\lambda}_-^{k_l}(p_l)\big)}{2 \omega_{p_l} (1 + e^{-\beta \omega_{p_l}})} \hat{\Psi}_B(-P, P),
\end{align*}
where now $\overline{\lambda}^{k_l}_{\pm}(p_l)$ is expressed in terms of $v_i$ and $\hat{\Psi}_B(-P, P)$ is as $\hat{\Psi}(-P, P)$ just with the $B_i \leftrightarrow A_i$. But, now $v_0 \leq v_1 \leq v_2 \leq \ldots \leq v_n$ and from the second condition in \eqref{eq: u-conditions} follows:
\begin{equation*}
    \beta - v_n = \beta - (\beta + u_{m-1} - u_m) = u_{m} - u_{m-1} \geq \frac{\beta}{n+1}.
\end{equation*}
Therefore, it is enough to study $F^{\beta}_{n,G}(U, \mathbf{Z})$ assuming the first condition in \eqref{eq: u-conditions} to hold.\\
The lines $l$ are denoted respectively $l_+$ and $l_-$, when among the $\lambda^{k_l}_{\pm}$ factors in \eqref{eq: Mos} the contribution with four-momentum is restricted respectively in the future or in the past light cone. Let, for this purpose:
\begin{equation*}
    \epsilon: E(G) \to \{+,-\}
\end{equation*}
be the map associating to each $l \in E(G)$ respectively $l_+$ or $l_-$. Denoting $E_{\pm}(G) = \{ l \in G : \epsilon(l) = \pm \}$ we have:
\begin{align}
    F^{\beta}_{n,G}(U, \mathbf{Z}) = \sum_{\epsilon,\mathbf{K}} \int \di P &\bigg(\prod_{l_+ \in E_+(G)} e^{ik_{l_+}\mathbf{p}_{l_+}  (\mathbf{z}_{s(l_+)} - \mathbf{z}_{r(l_+)})} \frac{\lambda^{k_{l_+}}_{\epsilon_+}(p_{l_+})}{2 \omega_{p_{l_+}} (1 + e^{-\beta \omega_{p_{l_+}}})}\bigg) \nonumber\\
    \times &\bigg( \prod_{l_- \in E_-(G)} e^{ik_{l_-}\mathbf{p}_{l_-} (\mathbf{z}_{s(l_-)} - \mathbf{z}_{r(l_-)})} \frac{-\lambda^{k_{l_-}}_{\epsilon_-}(p_{l_-})}{2 \omega_{p_{l_-}} (1 + e^{-\beta \omega_{p_{l_-}}})} \bigg)\hat{\Psi}(-P, P), \label{eq: espansione}
\end{align}
where $\epsilon_{\pm}$ is the corresponding choice for the lower index of $\lambda^{k_l}_{\pm}$ to have respectively future or past pointing momentum associated to the lines $l_+$ and $l_-$.
By its definition, $\hat{\Psi}(-P,P)$ is an entire analytic function that grows at most polynomially in each direction.
In particular we want to argue that:
\begin{equation}\label{eq: fastdecay}
    \int \di P^0  \bigg(\prod_{l_{\pm} \in E_{\pm}(G)}\frac{\lambda^{k_{l_+}}_{\epsilon_+}(p_{l_+})}{1 + e^{-\beta \omega_{p_{l_+}}}} \frac{\lambda^{k_{l_-}}_{\epsilon_-}(p_{l_-})}{1 + e^{-\beta \omega_{p_{l_-}}}} \bigg) \hat{\Psi}(-P, P),
\end{equation}
as a function of the three momenta $\mathbf{p}_l$, is of fast decay. Equation \eqref{eq: u-conditions} gives:
\begin{equation*}
    \max_{i < j}(u_j - u_i) = u_n - u_0  \leq \frac{n \beta}{n+1} < \beta.
\end{equation*}
Therefore:
\begin{equation*}
    e^{-\omega_{p_{l}} \beta} e^{\omega_{p_{l}}(u_{r(l)} - u_{s(l)})} = e^{-\omega_{p_{l}} (\beta - \frac{n\beta}{n+1})} e^{\omega_{p_{l}}(u_{r(l)} - u_{s(l)} - \frac{n\beta}{n+1})} 
\end{equation*}
and since $u_{r(l)} - u_{s(l)} - \frac{n\beta}{n+1} \leq 0$, we have:
\begin{equation*}
    e^{-\omega_{p_{l}} \beta} e^{\omega_{p_{l}}(u_{r(l)} - u_{s(l)})} \leq e^{-\omega_{p_{l}} \frac{\beta}{n+1}}.
\end{equation*}
Therefore, combining this with the result of the result of Proposition \ref{prop: 1}, we have that \eqref{eq: fastdecay} is of fast decay both in $\mathbf{p}_{l_-}$ and in $\mathbf{p}_{l_+}$. After these considerations, we distinguish the $m > 0$ and $m = 0$ cases.\\\\
$\bullet$ \textbf{Case $\mathbf{m>0}$}. We notice that:
\begin{equation}\label{eq: FermiFactor}
    \frac{1}{1 + e^{-\beta \omega_{p_l}}} = \sum_{j=0}^{\infty} (-1)^j \big( e^{-\beta \omega_{p_l}} \big)^j,
\end{equation}
and rewrite $F_{n,G}^{\beta}$ as:
\begin{align*}
    F^{\beta}_{n,G}(U, \mathbf{Z}) = \sum_{\epsilon,\mathbf{K}} \int \di P &\bigg(\prod_{l_+ \in E_+(G)} \sum_{j_1=0}^{\infty} (-1)^{j_1} e^{-j_1 \omega_{p_{l_+}}\beta} e^{ik_{l_+}\mathbf{p}_{l_+} (\mathbf{z}_{s(l_+)} - \mathbf{z}_{r(l_+)})} \frac{\lambda^{k_{l_+}}_+(p_{l_+})}{2 \omega_{p_{l_+}}}\bigg)\\
    \times &\bigg( \prod_{l_- \in E_-(G)} \sum_{j_2=0}^{\infty} (-1)^{j_2} e^{-j_2 \omega_{p_{l_-}}\beta} e^{ik_{l_-}\mathbf{p}_{l_-} (\mathbf{z}_{s(l_-)} - \mathbf{z}_{r(l_+)})} \frac{-\lambda^{k_{l_-}}_-(p_{l_-})}{2 \omega_{p_{l_-}}} \bigg)\hat{\Psi}(-P, P).
\end{align*}
Fix now $\varepsilon$, $\mathbf{K}$ and any term in the $j$-sums. Then, by the statement of Lemma \ref{lem: 1} we have:
\begin{equation*}
    |F_{n,G,\varepsilon, \mathbf{K}, j_1,j_2}^{\beta}(U, \mathbf{Z})| \leq c_{A_0, \ldots, A_n} \prod_{l \in E(G)} e^{-m\sqrt{|\mathbf{z}_{r(l)} - \mathbf{z}_{s(l)}|^2 + (\beta j_l)^2}},
\end{equation*}
where $j_l$ denotes either a $j_1$ or $j_2$ index depending on whether $l$ is a positive or negative line. Moreover, we used the finite range of $u_i$ and that $\frac{n \beta}{n+1} - u_{r(l)} + u_{s(l)} \geq 0$ to get rid of the $u$-dependence. Now, reintroducing the sum over $j_l$:
\begin{align*}
    \sum_{j=0}^{\infty} e^{-m\sqrt{q^2 + (\beta j)^2}} = \sum_{\beta j < q} e^{-m\sqrt{q^2 + (\beta j)^2}} + \sum_{\beta j \geq q} e^{-m\sqrt{q^2 + (\beta j)^2}}
\end{align*}
for $q>0$. The two sums separately give:
\begin{align*}
    \sum_{j = [q/\beta]}^{\infty} e^{-m\sqrt{q^2 + (\beta j)^2}} 
    &\leq e^{-mq} \sum_{j = 0}^{\infty} e^{-m \beta j} = \frac{e^{-mq}}{1 - e^{-m\beta}}\\
    \sum_{j = 0}^{[q/\beta]-1} e^{-m\sqrt{q^2 + (\beta j)^2}} &\leq \sum_{j = 0}^{[q/\beta] - 1} e^{-m\sqrt{q^2}} = \bigg[\frac{q}{\beta}\bigg]e^{-m q}. 
\end{align*} 
In the above, $[x]$ denotes the next approximation by an integer of $x \in \mathbb{R}$. Combining the two:
\begin{equation*}
    \sum_{j=0}^{\infty} e^{-m\sqrt{q^2 + (\beta j)^2}} \leq c' e^{-mq}
\end{equation*}
for some constant $c' \in \mathbb{R}^+$. From this we conclude that:
\begin{equation*}
    |F_{n,G,\varepsilon, \mathbf{K}}^{\beta}(U, \mathbf{Z})| \leq c_{A_0, \ldots, A_n} \prod_{l \in E(G)} e^{-m|\mathbf{z}_{r(l)} - \mathbf{z}_{s(l)}|}
\end{equation*}
As very last step, the connectedness of the graph ensures:
\begin{equation*}
    \sum_{l \in E(G)}|\mathbf{z}_{r(l)} - \mathbf{z}_{s(l)}| \geq \max_{i \in \{ 1 , \ldots, n\}} |\mathbf{z}_i| \geq \sqrt{\frac{1}{n} \sum_{i=0}^n |\mathbf{z}_i|^2} := \frac{1}{\sqrt{n}} r_e.
\end{equation*}
Follows that: 
\begin{equation*}
     |F_{n,G,\varepsilon, \mathbf{K}}^{\beta}(U, \mathbf{Z})| \leq c_{A_0, \ldots, A_n} e^{-\frac{m}{\sqrt{n}}r_e},
\end{equation*}
and the prove is concluded as $F_{n}^{\beta}$ is just a sum of finitely many $F_{n,G,\varepsilon, \mathbf{K}}^{\beta}(U, \mathbf{Z})$.\\\\
$\bullet$ \textbf{Case $\mathbf{m = 0}$}. First of all we notice that the fast decay property of the integrand is not affected by setting $m=0$ as the Fermi factor does not introduce additional singularities. Moreover, using again the series representation of the Fermi factor \eqref{eq: FermiFactor}, we need to estimate:
\begin{align*}
    F^{\beta}_{n,G}(U, \mathbf{Z}) = \sum_{\epsilon,\mathbf{K}} \int \di P &\bigg(\prod_{l_+ \in E_+(G)} \sum_{j_1=0}^{\infty} (-1)^{j_1} e^{-j_1 |\mathbf{p}_{l_+}|\beta} e^{ik_{l_+}\mathbf{p}_{l_+} (\mathbf{z}_{s(l_+)} - \mathbf{z}_{r(l_+)})} \frac{\lambda^{k_{l_+}}_+(p_{l_+})}{2 |\mathbf{p}_{l_+}|}\bigg)\\
    \times &\bigg( \prod_{l_- \in E_-(G)} \sum_{j_2=0}^{\infty} (-1)^{j_2} e^{-j_2 |\mathbf{p}_{l_-}|\beta} e^{ik_{l_-}\mathbf{p}_{l_-} (\mathbf{z}_{s(l_-)} - \mathbf{z}_{r(l_+)})} \frac{-\lambda^{k_{l_-}}_-(p_{l_-})}{2 |\mathbf{p}_{l_-}|} \bigg)\hat{\Psi}(-P, P).
\end{align*}
Therefore, by Lemma \ref{lem: 1}, fixing $\epsilon$, $\mathbf{K}$ and any term in the $j$-sums we get:
\begin{equation*}
    |F_{n,G,\varepsilon, \mathbf{K}, j_1,j_2}^{\beta}(U, \mathbf{Z})| \leq c_{A_0, \ldots, A_n} \prod_{l \in E(G)}\frac{1}{\big(1 + \sqrt{(u_{s(l)} - u_{r(l)})^2 + |\mathbf{z}_{s(l)} - \mathbf{z}_{r(l)} |^2} \big)^3}.
\end{equation*}
Now, since $u_{s(l)} - u_{r(l)} \geq - \frac{n \beta}{n+1} + \min_{l \in E(G)}(c_{\beta} - u_{r(l)} + u_{s(l)}) \coloneqq c'_{\beta}$, the following estimate holds:
\begin{equation*}
    |F_{n,G,\varepsilon, \mathbf{K}, j_1,j_2}^{\beta}(U, \mathbf{Z})| \leq c_{A_0, \ldots, A_n} \prod_{l \in E(G)}\frac{1}{\big(1 + |\mathbf{z}_{s(l)} - \mathbf{z}_{r(l)} | \big)^3}.
\end{equation*}
Lastly, summing over all possible multiindices $\mathbf{K}$, combinations $\epsilon$ and connected graphs with $n+1$ vertices (all finite sums) we have the claim.
\end{proof}

\begin{prop}
Let $\omega^{\infty}$ be the ground state with respect to $\tau_t$ on $\mathfrak{A}$, the algebra of smeared Wick polynomials, with translation invariant two point functions $\omega^{\infty,+}_2(x,y)$ and $\omega^{\infty,-}_2(x,y)$. Then, denoting by $\mathcal{S}_n^{\infty} = \{ (u_1,\ldots, u_n) \in \mathbb{R}^n: -\infty < u_1 \leq \ldots \leq u_n < \infty \}$ and for all $A_0, \ldots, A_n \in \mathfrak{A}(\mathcal{O})$ with $\mathcal{O} \subset B_R \subset \mathbb{R}^4$, the connected correlation functions:
\begin{equation}
    F_{n}^{\infty}(u_1, \mathbf{z}_1; \ldots; u_n, \mathbf{z}_n) = \omega^{\infty, \mathcal{T}}(\alpha_{iu_1, \mathbf{z}_1} (A_1) \otimes \ldots \otimes A_0 \otimes \ldots \otimes \alpha_{iu_n, \mathbf{z}_n} (A_n))
\end{equation}
\begin{itemize}
    \item Decay exponentially in $\mathcal{S}_n^{\infty} \times \mathbb{R}^{3n}$ if $m > 0$ and $r_g > 2R$:
    \begin{equation*}
        |F_n^{\infty}(u_1, \mathbf{z}_1; \ldots; u_n, \mathbf{z}_n)| \leq c_{A_0, \ldots, A_n} e^{-\frac{m}{\sqrt{n}}r_g} \, , \quad r_g = \sqrt{\sum_{i=1}^n u_i^2 + |\mathbf{z}_i^2|},
    \end{equation*}
    for some $c_{A_0, \ldots, A_n} \in \mathbb{R}$ depending on the chosen observables.
    \item Decay in $\mathcal{S}_n^{\infty} \times \mathbb{R}^{3n}$ if $m = 0$:
    \begin{equation*}
        |F_n^{\infty}(u_1, \mathbf{z}_1; \ldots; u_n, \mathbf{z}_n)| \leq c_{A_0, \ldots, A_n} \sum_{G \in \mathcal{G}^{c}_{n+1}}\prod_{l \in E(G)}\frac{1}{\bigg(1 + \sqrt{(u_{s(l)} - u_{r(l)})^2 + | \mathbf{z}_{s(l)} - \mathbf{z}_{r(l)} |^2}\bigg)^{3}}
    \end{equation*}
    for some $c_{A_0, \ldots, A_n} \in \mathbb{R}$ depending on the chosen observables. Here $G$ denotes a graph, with edges $l$ of source $s(l)$ and range $r(l)$ in the set of connected graphs with $n+1$ vertices.
\end{itemize}
\end{prop}
\begin{proof}\label{app: proof2}
The proof is similar to that of \ref{thm: KMS}. Writing everything as a sum over connected graphs and focusing on a single $G \in \mathcal{G}^c_{n+1}$:
\begin{equation*}
    F_{n,G}^{\infty}(u_1, \mathbf{z}_1; \ldots; u_n, \mathbf{z}_n) = \sum_{\mathbf{K}} \int \di X \, \di Y \prod_{l \in E(G)} \omega_2^{\infty,k_l}(x_l,y_l) \Psi^{\mathbf{K}}(X,Y)\bigg|_{\mathrm{Conf} = 0},
\end{equation*}
where $\mathbf{K} = (k_1, \ldots, k_N)$ is a multiindex with $N = |E(G)|$ arguments, each of which is $k_l = \pm$, $dX = dx_1 \cdots dx_N$. Here, the two point functions, are those of the ground state for the free theory:
\begin{align*}
    \omega^{\infty, +}_{2}(x-y) &= \frac{1}{(2 \pi)^3} \int \frac{\di^3\mathbf{p}}{2 \omega_p}(-\gamma^0 \omega_p - \gamma^i p_i + m)e^{-i\omega_p (t_x - t_y)} e^{i \mathbf{p} (\mathbf{x}- \mathbf{y})}\\
    \omega^{\infty, -}_{2}(x-y) &= -\frac{1}{(2 \pi)^3} \int \frac{\di^3\mathbf{p}}{2 \omega_p} (\gamma^0 \omega_p - \gamma^i p_i + m)e^{i\omega_p (t_x - t_y)} e^{i \mathbf{p} (\mathbf{x}- \mathbf{y})}
\end{align*}
where in $\omega_2^{\infty,k_l}(x_l,y_l)$ the order of the difference depends on $k_l$ correspondingly with the definition of the product $\star_{\omega^{\infty}}$, and:
\begin{equation*}
    \Psi^{\mathbf{K}}(X,Y) \coloneqq \prod_{l \in E(G)} \mathfrak{D}^{k_l}(x_l,y_l) (\alpha_{iu_1, \mathbf{z}_1} (A_1) \otimes \ldots \otimes A_0 \otimes \ldots \otimes \alpha_{iu_n, \mathbf{z}_n} (A_n))
\end{equation*}
for $\mathfrak{D}^{k_l}(x_l,y_l)$ the functional derivatives as they appear in \eqref{eq: Gamma}.\\
Moving now the spacetime translations to the vacuum two point function and writing all in momentum space:
\begin{equation*}
    F_{n,G}^{\infty}(u_1, \mathbf{z}_1; \ldots; u_n, \mathbf{z}_n) = \sum_{\mathbf{K}} \int \di P \prod_{l \in E(G)} k_l \frac{e^{i\mathbf{p}_l(\mathbf{z}_{s(l)} - \mathbf{z}_{r(l)})}}{2 \omega_{p_l}} e^{\omega_{p_l}(u_{s(l)} - u_{r(l)})} (-\cancel{p}_l + m) \delta(p_l^0 - k_l \omega_{p_l}) \hat{\Psi}(-P,P),
\end{equation*}
where $s(l) < r(l)$ are the source and range vertices of each line $l \in E(G)$.\\
Therefore, by definition of $\mathcal{S}^{\infty}_n$, the exponential factor is always decreasing in $|\mathbf{p}_l|$. Moreover, computing the integrals involving the Dirac deltas, $(p_1, \ldots, p_n) \mapsto \hat{\Psi}(-P,P)$ has momenta in the upper cone $(p_1, \ldots, p_n) \in (V^+)^n$. Therefore, by Proposition \ref{prop: 1}, $\hat{\Psi}$ is of fast decay. Given these considerations, we distinguish the $m>0$ from the $m=0$ case.\\\\
$\bullet$ \textbf{Case $\mathbf{m > 0}$}. Lemma \ref{lem: 1} gives for:
\begin{equation*}
    r = \sum_{l \in E(G)}\sqrt{(u_{r(l)} - u_{s(l)})^2 + |\mathbf{z}_{r(l)} - \mathbf{z}_{s(l)}|^2},
\end{equation*}
and each fixed $\mathbf{K}$:
\begin{equation*}
    |F_{n,G,\mathbf{K}}^{\infty}(u_1, \mathbf{z}_1; \ldots; u_n, \mathbf{z}_n)| \leq c_{A_0, \ldots, A_n} e^{-m r} \, \quad \mathrm{for} \,\, c_{A_0, \ldots, A_n} \in \mathbb{R}^+.
\end{equation*}
As every graph is connected and thus each vertex can be reached starting from $(0, \mathbf{0})$, we have by triangular inequality:
\begin{equation*}
    \sum_{l \in E(G)} \sqrt{(u_{r(l)} - u_{s(l)})^2 + |\mathbf{z}_{r(l)} - \mathbf{z}_{s(l)}|^2} \geq \max_{i \in \{ 0, \ldots, n \}} \sqrt{u_i^2 + |\mathbf{z}_i|^2}.
\end{equation*}
Therefore, averaging:
\begin{equation*}
    r \geq \sqrt{\frac{1}{n} \sum_{i =1}^n u_i^2 + |\mathbf{z}_i|^2} := \frac{1}{\sqrt{n}} r_g.
\end{equation*}
Follows that:
\begin{equation*}
    |F_{n,G, \mathbf{K}}^{\infty}(u_1, \mathbf{z}_1; \ldots; u_n, \mathbf{z}_n)| \leq c_{A_0, \ldots, A_n} e^{-\frac{m}{\sqrt{n}} r_g},
\end{equation*}
and summing over all possible, finitely many, connected graphs with $n+1$ vertices and multiindices $\mathbf{K}$ we have the desired estimate.\\\\
$\bullet$ \textbf{Case $\mathbf{m = 0}$}. Still by Lemma \ref{lem: 1}, the following estimate holds for fixed multiindex $\mathbf{K}$:
\begin{equation*}
        |F_{n,G,\mathbf{K}}^{\infty}(u_1, \mathbf{z}_1; \ldots; u_n, \mathbf{z}_n)| \leq c_{A_0, \ldots, A_n} \sum_{G \in \mathcal{G}^{c}_{n+1}}\prod_{l \in E(G)}\frac{1}{\bigg(1 + \sqrt{(u_{s(l)} - u_{r(l)})^2 + | \mathbf{z}_{s(l)} - \mathbf{z}_{r(l)} |^2}\bigg)^{3}}
    \end{equation*}
Finally, summing over all graphs $G$ and multiindices $\mathbf{K}$ we have the claim.
\end{proof}

\section{Technical Propositions and Lemmata}
\begin{lem}\label{lem: 2}
The following integral is convergent:
\begin{align*}
    \int_{\mathbb{R}^3 \times \mathbb{R}^3} \di^3\mathbf{x} \di^3\mathbf{y} \, \frac{1}{(1 + |\mathbf{x}|)^3}\frac{1}{(1 + |\mathbf{y}|)^3} \frac{1}{(1 + |\mathbf{x} - \mathbf{y}|)^3} < \infty
\end{align*}
\end{lem}
\begin{proof}
We write:
\begin{equation*}
    \int_{\mathbb{R}^3 \times \mathbb{R}^3} \di^3 \mathbf{x} \di^3 \mathbf{y} \, \frac{1}{(1+ |\mathbf{x}|)^3}\frac{1}{(1+ |\mathbf{y}|)^3} \frac{1}{(1+  |\mathbf{x} - \mathbf{y}|)^3} = \int_{\mathbb{R}^3} \di^3\mathbf{x} \frac{1}{(1+ |\mathbf{x}|)^3} \int_{\mathbb{R}^3} \di^3\mathbf{y} \frac{1}{(1+ |\mathbf{y}|)^3} \frac{1}{(1+  |\mathbf{x} - \mathbf{y}|)^3}.
\end{equation*}
Let us start computing the innermost integral and, for that purpose, we divide the region of integration in three parts:
\begin{itemize}
    \item Region $\mathbf{I}$ corresponding to $|\mathbf{y}| \geq 2 |\mathbf{x}|$
    \item Region $\mathbf{II}$ corresponding to $|\mathbf{y}| \leq \frac{1}{2} |\mathbf{x}|$
    \item Region $\mathbf{III}$ corresponding to $\frac{1}{2}|\mathbf{x}| \leq |\mathbf{y}| \leq 2 |\mathbf{x}|$
\end{itemize}
We start focusing on region $\mathbf{I}$. In this region we have, by the triangular inequality:
\begin{equation*}
    |\mathbf{x} - \mathbf{y}| \geq |\mathbf{y}| - |\mathbf{x}| \geq \frac{1}{2}|\mathbf{y}|,
\end{equation*}
but also:
\begin{equation*}
    |\mathbf{x} - \mathbf{y}| \geq |\mathbf{x}|.
\end{equation*}
Then, we can perform the following estimate:
\begin{align*}
    \int_{|\mathbf{y}| \geq 2 |\mathbf{x}|} \di^3\mathbf{y} \frac{1}{(1+ |\mathbf{y}|)^3} \frac{1}{(1+  |\mathbf{x} - \mathbf{y}|)^{3}} 
    &\leq \frac{1}{(1+  |\mathbf{x}|)^{3/2}} \int_{|\mathbf{y}| \geq 2 |\mathbf{x}|} \di^3\mathbf{y} \frac{1}{(1+ \frac{1}{2}|\mathbf{y}|)^{9/2}}\\
    &= \frac{4 \pi}{(1+  |\mathbf{x}|)^{3/2}} \bigg( \frac{2}{3}\frac{1}{(1 + |\mathbf{x}|)^{3/2}} + \frac{2}{7}\frac{1}{(1 + |\mathbf{x}|)^{7/2}} - \frac{4}{5}\frac{1}{(1 + |\mathbf{x}|)^{5/2}} \bigg).
\end{align*}
Therefore, by picking an arbitrary small $\epsilon \in (0,3)$, we have shown: 
\begin{align*}
    \bigg| \frac{4 \pi}{(1+  |\mathbf{x}|)^{3/2}} \bigg( \frac{2}{3}\frac{1}{(1 + |\mathbf{x}|)^{3/2}} + \frac{2}{7}\frac{1}{(1 + |\mathbf{x}|)^{7/2}} - \frac{4}{5}\frac{1}{(1 + |\mathbf{x}|)^{5/2}} \bigg) \bigg| \leq \frac{C_{\mathbf{I}}}{(1 + |\mathbf{x}|)^{3-\epsilon}},
\end{align*}
where $C_{\mathbf{I}}$ is a non vanishing constant.\\
Now we look at region $\mathbf{II}$. Again by triangular inequality, follows:
\begin{equation*}
    |\mathbf{x} - \mathbf{y}| \geq \frac{1}{2}|\mathbf{x}|.
\end{equation*}
Then:
\begin{align*}
    \int_{|\mathbf{y}| \leq \frac{1}{2} |\mathbf{x}|} \di^3\mathbf{y} \frac{1}{(1+ |\mathbf{y}|)^3} \frac{1}{(1+  |\mathbf{x} - \mathbf{y}|)^{3}} &\leq \frac{1}{(1+ \frac{1}{2} |\mathbf{x}|)^{3}}\int_{|\mathbf{y}| \leq \frac{1}{2} |\mathbf{x}|} \di^3\mathbf{y} \frac{1}{(1+ |\mathbf{y}|)^3}\\
    &= \frac{4 \pi}{(1+ \frac{1}{2} |\mathbf{x}|)^{3}} \bigg( \frac{1}{2} \ln\bigg( 1 + \frac{1}{2}|\mathbf{x}| \bigg) + \frac{5}{2} - \frac{1}{1 + \frac{1}{2}|\mathbf{x}|} \bigg( 2 + \frac{1}{2 + |\mathbf{x}|} \bigg)\bigg).
\end{align*}
Now, for any choice of $\epsilon \in (0,3)$, the above can be estimated as:
\begin{align*}
    \bigg| \frac{4\pi}{(1+ \frac{1}{2} |\mathbf{x}|)^{3}} \bigg( \frac{1}{2} \ln\bigg( 1 + \frac{1}{2}|\mathbf{x}| \bigg) + \frac{5}{2} - \frac{1}{1 + \frac{1}{2}|\mathbf{x}|} \bigg( 2 + \frac{1}{2 + |\mathbf{x}|} \bigg)\bigg) \bigg| 
    &\leq \frac{C_{\mathbf{II}}}{\big( 1 + \frac{1}{2}|\mathbf{x}|\big)^{3-\epsilon}} \sim \frac{C_{\mathbf{II}}}{\big( 1 +|\mathbf{x}|\big)^{3-\epsilon}}
\end{align*}
where $C_{\mathbf{II}}$ is a non vaninshing constant.\\
Finally, in region $\mathbf{III}$, we notice that:
\begin{equation*}
    \frac{1}{(1+ |\mathbf{y}|)^3} \sim \frac{1}{(1+ |\mathbf{x}|)^3}.
\end{equation*}
Therefore:
\begin{align*}
    \int_{\frac{1}{2} |\mathbf{x}| \leq |\mathbf{y}| \leq 2 |\mathbf{x}|} \di^3\mathbf{y} \frac{1}{(1+ |\mathbf{y}|)^3} \frac{1}{(1+  |\mathbf{x} - \mathbf{y}|)^{3}} &\sim  \frac{1}{(1+ |\mathbf{x}|)^3} \int_{\frac{1}{2} |\mathbf{x}| \leq |\mathbf{y}| \leq 2 |\mathbf{x}|} \di^3\mathbf{y} \frac{1}{(1+  |\mathbf{x} - \mathbf{y}|)^{3}}\\
    &\leq  \frac{1}{(1+ |\mathbf{x}|)^3} \int_{|\mathbf{v}| \leq 3 |\mathbf{x}|} \di^3\mathbf{v} \frac{1}{(1+  |\mathbf{v}|)^{3}}\\
    &= \frac{4\pi}{(1+ |\mathbf{x}|)^3}  \bigg( \frac{1}{2} \ln\bigg( 1 + 3|\mathbf{x}| \bigg) + \frac{5}{2} - \frac{1}{1 + 3|\mathbf{x}|} \bigg( 2 + \frac{1}{2 + 6|\mathbf{x}|} \bigg)\bigg),
\end{align*}
where we have substituted $\mathbf{v} = \mathbf{x} - \mathbf{y}$ and the positivity of the integrand together with:
\begin{equation*}
    \{ \mathbf{v} \in \mathbb{R}^3 : |\mathbf{x} - \mathbf{v}| \leq 2 |\mathbf{x}|\} \subset  \{ \mathbf{v} \in \mathbb{R}^3 : |\mathbf{v}| \leq 3 |\mathbf{x}|\}.
\end{equation*}
Follows that, for any choice of $\epsilon \in (0,3)$:
\begin{align*}
    \bigg| \frac{4\pi}{(1+ |\mathbf{x}|)^3}  \bigg( \frac{1}{2} \ln\bigg( 1 + 3|\mathbf{x}| \bigg) + \frac{5}{2} - \frac{1}{1 + 3|\mathbf{x}|} \bigg( 2 + \frac{1}{2 + 6|\mathbf{x}|} \bigg)\bigg) \bigg| &\leq \frac{C_{\mathbf{III}}}{(1 + 3|\mathbf{x}|)^{3-\epsilon}} \sim \frac{C_{\mathbf{III}}}{(1 + |\mathbf{x}|)^{3-\epsilon}},
\end{align*}
where $C_{\mathbf{III}}$ is a non vanishing constant.\\
Therefore, for each region, calling $C = \max\{ C_{\mathbf{I}}, C_{\mathbf{II}}, C_{\mathbf{III}} \}$ the remaining integration in $\mathbf{x}$ gives the final result:
\begin{equation*}
    C\int_{\mathbb{R}^3} \di^3\mathbf{x} \frac{1}{(1+ |\mathbf{x}|)^3} \frac{1}{(1+ |\mathbf{x}|)^{3-\epsilon}} < \infty,
\end{equation*}
proving the claim.
\end{proof}
\begin{rem}
The above is the proof for two variables $\mathbf{x},\mathbf{y}$. However, by a sufficiently small choice of $\epsilon$, the proof extends to arbitrary many variables by iteratively reducing the degree of the denominator of a $\epsilon$ factor.
\end{rem}
\begin{prop}\label{prop: 1}
Let $A_0, \ldots, A_n \in \mathfrak{A}$ be fermionic functionals and $d \in \mathbb{N}$. Then, the compactly supported distribution:
\begin{equation*}
    \Psi^{\mathbf{K}}(x_1, \ldots, x_d, y_1, \ldots, y_d) = \prod_{l=1}^d \mathfrak{D}^{k_l}(x_l,y_l)(A_0 \otimes A_1 \otimes \ldots \otimes A_n)\bigg|_{\psi_0 = \cdots = \psi_n = 0},
\end{equation*}
where $\mathbf{K} = (k_1, \ldots, k_d)$ is a multiindex with entries $k_i \in \{ +,-\}$, and:
\begin{equation*}
    \mathfrak{D}^{+}(x_l,y_l) = \frac{\delta_r}{\delta \psi(x_l)} \otimes \frac{\delta}{\delta \overline{\psi}(y_l)} \, , \qquad    \mathfrak{D}^{-}(x_l,y_l) = \frac{\delta_r}{\delta \overline{\psi}(x_l)} \otimes \frac{\delta}{\delta \psi(y_l)},
\end{equation*}
has Fourier transform $(p_1, \ldots, p_d) \mapsto \hat{\Psi}^{\mathbf{K}}(-k_1 p_1, \ldots, -k_d p_d, k_1 p_1, \ldots, k_d p_d)$ of fast decay in a neighbourhood of the $d$-fold product of the forward light cone $(V^+)^d$ with that of the backward light cone $(V^-)^d$
\end{prop}

\begin{proof}
The proof is analogous as Proposition 8 of \cite{FredenhagenLindnerKMS_2014} or Proposition D.1 of \cite{BrunettiFredenhagenNic}, with the only difference that the $A_i$, for all $i = 1, \ldots,n$, are spinorial microcausal functionals. This implies that, component-wise, they are microcausal functionals. Choosing a trivializing basis $\{ e_C\}_{C = 1, \ldots, 4}$, of the spinor bundle over Minkowski spacetime write:
\begin{equation*}
    A_i = \sum_{C = 1}^4 A_i^C e_C
\end{equation*}
Moreover, for any two distributions $u,v \in \mathcal{D}'(\mathbb{R}^4)$:
\begin{equation*}
    \mathrm{WF}(u+v) \subseteq \mathrm{WF}(u) \cup \mathrm{WF}(v)
\end{equation*}
So, by the microcausality of the functional, $\mathrm{WF}(A_i^C) \subset  \left\{ (x_1, \ldots, x_n; p_1, \ldots, p_n) \in \mathbb{R}^{4n}\times \Dot{\mathbb{R}}^{4n} \, : \, \sum_{k=1}^{n}p_k = 0 \right\}$ for each component $C$. The rest of the proof is analogous as the above cited references.
\end{proof}

\begin{lem}\label{lem: 1}
Let $f \in \mathcal{E}'(\mathbb{R}^4, \mathbb{C}^4)$, with $\mathrm{supp}(f) \subset B_R \subset \mathbb{R}^4$ whose Fourier transform is rapidly decaying in an open neighbourhood of the positive and negative lightcone. Then, the functions:
\begin{equation*}
    I^{\sigma, s; \mu}_{f}(x^0, \mathbf{x}) = \int \frac{\di^3\mathbf{p}}{2 \omega_p} e^{-i (\sigma x^0 \omega_p - s \mathbf{p} \mathbf{x})} \hat{f}(\omega_p, \mathbf{p})^{\nu}(m + (\sigma \gamma^0 \omega_p - \gamma^i p_i))_{\nu}^{\,\,\,\, \mu}
\end{equation*}
where $s = \pm$, $\sigma = \pm$ and $\omega_p = \sqrt{\mathbf{p}^2 + m^2}$, have an analytic continuation into the region $\mathbb{C}_{-\sigma} \times \mathbb{R}^{3}$ where $\mathbb{C}_{-\sigma}$ denotes respectively the upper or lower half plane respectively for $\sigma = -$ or $\sigma = +$. Moreover, if $m > 0$ and $r > \sqrt{2}R$, it holds that (for each $\mu$):
\begin{equation*}
    |I^{\sigma, s; \mu}_{f}(-\sigma iu, \mathbf{x})| \leq c_1 e^{-m r},
\end{equation*}
for $r = \sqrt{u^2 + \mathbf{x}^2}$ and for $c_1 >0$ depending on $f$. While, for $m=0$ we have that (for each $\mu$):
\begin{equation*}
    |I^{\sigma, s; \mu}_{f}(-\sigma iu, \mathbf{x})| \leq \frac{c'_1}{(1+r)^3}
\end{equation*}
where $c'_1 > 0$ depending on $f$, $\sigma u > 0$.
\end{lem}
\begin{proof}
The proof of the $m>0$ case, follows the same steps performed in Proposition $7$ of \cite{FredenhagenLindnerKMS_2014} that uses the Paley-Wiener Theorem to estimate the Fourier transform of $f$. The only precaution to take is that $f$ is vector valued. The integrand is of the kind:
\begin{equation*}
    \hat{f}(\omega_p,\mathbf{p})^{\nu} \cancel{p}_{\nu}^{\,\,\,\,\mu}, \qquad \hat{f}(\omega_p,\mathbf{p})^{\nu} m\delta_{\nu}^{\,\,\,\,\mu}
\end{equation*}
However, as componentwise $\hat{f}(\omega_p,\mathbf{p})^{\nu}$ is fastly decreasing, the multiplication by a constant or by a polynomial in the momentum does not change the fastly decreasing property. Therefore, both remain of fast decay and Proposition 7 in \cite{FredenhagenLindnerKMS_2014} ensures the validity of the statement componentwise.\\\\
Let us extensively discuss the $m=0$ case where $\omega_p = | \mathbf{p} |$. First notice that:
\begin{equation*}
    (\sigma\gamma^0 \omega_p - \gamma^i p_i)_{\nu}^{\,\,\,\, \mu} = |\mathbf{p}| \Gamma_{\nu}^{\,\,\,\, \mu},
\end{equation*}
where $\Gamma_{\nu}^{\,\,\,\, \mu}$ is the resulting combination of Dirac $\gamma$ matrices. Therefore, for each fixed $\nu = 0,1,2,3$, focus on:
\begin{equation*}
    I_{f}^{\sigma, s} = \frac{1}{2} \int \di^3\mathbf{p}e^{i \mathbf{p} \mathbf{x} s} e^{-i\sigma x^0 |\mathbf{p}|} \hat{f}^{\nu}(|\mathbf{p}|,\mathbf{p}).
\end{equation*}
The functions $\hat{f}(p)^{\nu}$ are, componentwise, polynomially bounded functions while the remaining integrands are:
\begin{equation*}
    I^{\sigma, s}_{f} : \frac{e^{i \mathbf{p} \mathbf{x} s -i \sigma \Re(z) |\mathbf{p}| + \sigma \Im(z) |\mathbf{p}|}}{2}\, ,
\end{equation*}
where we have denoted by $z$ the complex extension of $x^0$. Therefore, for $\sigma \Im(z) < 0$, the integral has analytic extension in $\mathbb{C}_{-\sigma} \times \mathbb{R}^{3}$. This proves the first statement.\\
For the second statement, the existence of the complex extension, leads to the computation for each fixed $\nu = 0,1,2,3$ of:
\begin{equation}\label{eq: integralepartenza}
    \frac{1}{2}\int \di^3\mathbf{p}e^{i \mathbf{p} \mathbf{x} s} e^{-\sigma u |\mathbf{p}|} \hat{f}^{\nu}(|\mathbf{p}|,\mathbf{p}), 
\end{equation}
where $\sigma u > 0$. Let us introduce a lower bound on the magnitude of momentum:
\begin{equation*}
    \frac{1}{2}\int \di^3\mathbf{p}e^{i \mathbf{p} \mathbf{x} s} e^{-\sigma u |\mathbf{p}|} \hat{f}^{\nu}(|\mathbf{p}|,p)  = \lim_{\chi \to 0} \frac{1}{2} \int_{\chi}^{\infty} \di |\mathbf{p}| \int_{0}^{\pi}\di \theta \sin \theta \int_{0}^{2 \pi} \di \varphi \, |\mathbf{p}|^2 e^{i \mathbf{p} \mathbf{x} s} e^{-\sigma u |\mathbf{p}|} \hat{f}^{\nu}(|\mathbf{p}|,\mathbf{p}).
\end{equation*}
Then, the exponential in $u$ is rewritten using that $|\mathbf{p}| > 0$:
\begin{equation*}
    \frac{1}{2 \pi} \int_{\mathbb{R}} dp_0 \frac{e^{i\sigma p_0 u}}{(p_0)^2 + |\mathbf{p}|^2} = \frac{e^{-\sigma|\mathbf{p}| u}}{2 |\mathbf{p}|} \hspace{20pt} \sigma u > 0.
\end{equation*}
This is proven by the aid of complex integration, using the Jordan Lemma and the Residue Theorem in the upper half plane. Therefore, the above integral becomes: 
\begin{equation*}
    = \lim_{\chi \to 0} \int d^4 p \frac{|\mathbf{p}|}{2 \pi} \frac{e^{i(\sigma p_0 u + \mathbf{p} \mathbf{x} s)}}{(p_0)^2 + |\mathbf{p}|^2}\hat{f}^{\nu}(|\mathbf{p}|,\mathbf{p}).
\end{equation*}
Let us now perform the change of coordinates:
\begin{equation*}
    \mathbf{x} = \mathbf{n} r \cos \alpha \, , \quad u = r \sin \alpha \, , \quad r = \sqrt{u^2 + |\mathbf{x}|^2},
\end{equation*}
where we choose $\mathbf{n} = (\sigma/s,0,0)$ and $0 < \alpha < \pi/2$. Correspondingly, we perform the change of coordinates also in momentum space:
\begin{equation*}
    \begin{pmatrix}
        k_0\\
        k_1\\
        k_2\\
        k_3
    \end{pmatrix} = \begin{pmatrix}
        \sin \alpha & \cos \alpha & 0 & 0\\
        \cos \alpha & -\sin \alpha & 0 & 0\\
        0 & 0 & 1 & 0\\
        0 & 0 & 0 & 1
    \end{pmatrix} \begin{pmatrix}
        p_0\\
        p_1\\
        p_2\\
        p_3
    \end{pmatrix}
\end{equation*}
From which it holds:
\begin{equation*}
    (p_0)^2 + |\mathbf{p}|^2 = (k_0)^2 + |\mathbf{k}|^2\, , \quad \mathbf{p}(\mathbf{k}) = (k_0 \cos \alpha - k_1 \sin \alpha, k_2, k_3),
\end{equation*}
and we insert this change of coordinates in the above integration:
\begin{equation}\label{eq: integraletutto}
    = \lim_{\chi \to 0} \int d^4 k \frac{|\mathbf{p}(\mathbf{k})|}{2 \pi} \frac{e^{i\sigma rk_0}}{(k_0)^2 + |\mathbf{k}|^2}\hat{f}^{\nu}(|\mathbf{p}(\mathbf{k})|,\mathbf{p}(\mathbf{k})).
\end{equation}
We distinguish the $r < 1$ and $r > 1$ cases.\\\\
For $r < 1$, by the fast decay property of $\hat{f}^{\nu}(|\mathbf{p}|,\mathbf{p})$ we know that for any $N \in \mathbb{N}$ it exist a constant $C_N \in \mathbb{R}$ such that for any fixed $\nu = 0,1,2,3$:
\begin{equation}\label{eq: fast decay}
    |\hat{f}^{\nu}(|\mathbf{p}|,\mathbf{p})| \leq \frac{C_N}{(1 + |\mathbf{p}|)^N}.
\end{equation}
Therefore, directly looking at \eqref{eq: integralepartenza}, estimating $|\hat{f}^{\nu}(|\mathbf{p}|,\mathbf{p})|$ with some choice of $N \in \mathbb{N}$ large enough:
\begin{equation*}
    \bigg|\frac{1}{2}\int \di^3\mathbf{p}e^{i \mathbf{p} \mathbf{x} s} e^{-\sigma u |\mathbf{p}|} \hat{f}^{\nu}(|\mathbf{p}|,\mathbf{p}) \bigg| \leq c \qquad r < 1
\end{equation*}
for some constant $c \in \mathbb{R}^+$ depending on $f$.\\\\
For $r > 1$, in order to compute \eqref{eq: integraletutto}, we displace $k_0$ in the complex plane calling it, from now on, $z$. The integral has two poles at $z = \pm i |\mathbf{k}|$ and two branch cuts at:
\begin{equation*}
    z = k_1 \tan \alpha \pm i \frac{\sqrt{k_2^2 + k_3^2}}{\cos \alpha}.
\end{equation*}
So, assuming without loss of generality that $\sigma = +$, we choose a contour of integration $\Gamma$ in the upper half plane (lower if instead $\sigma = -$) of the following shape:
\begin{figure}[H]
 		\centering
 		\includegraphics[width=1.10\columnwidth]{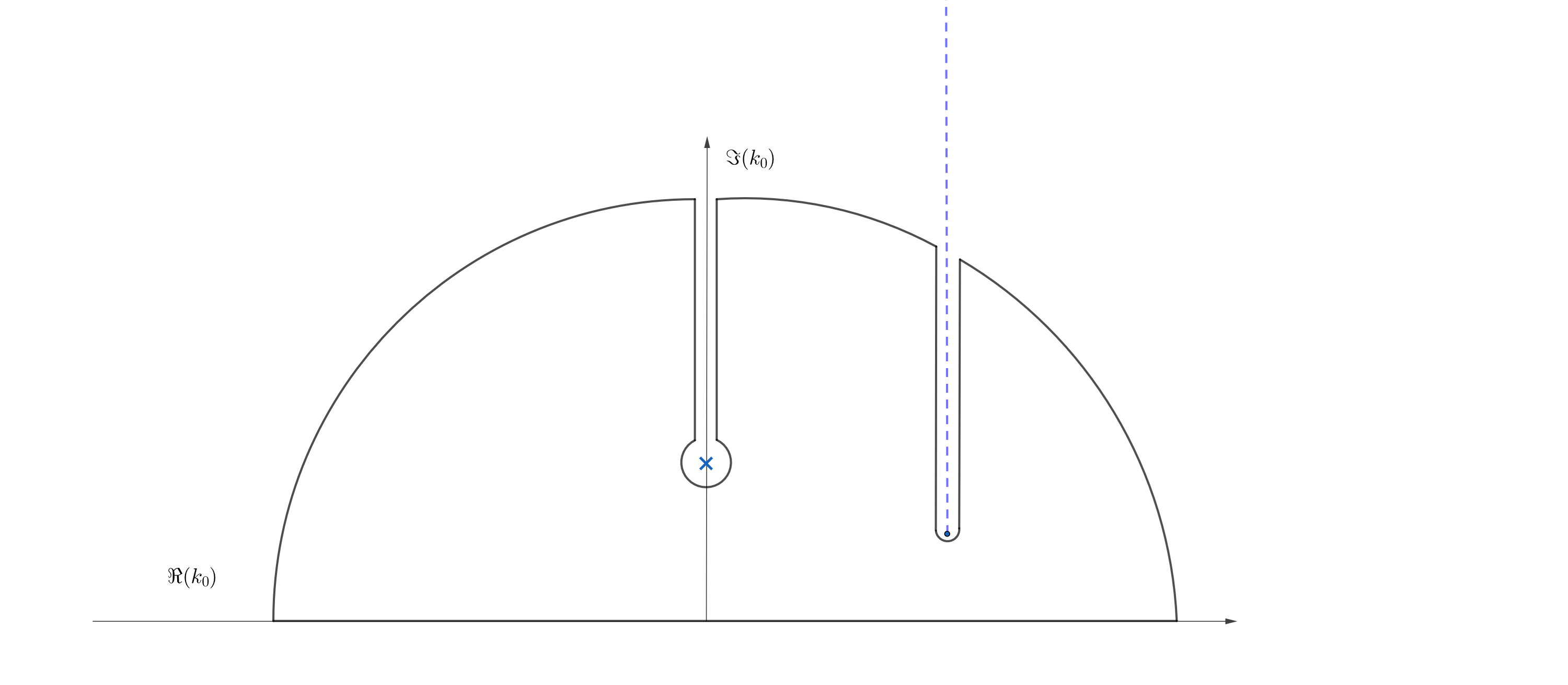}
 		\caption{The blue cross denotes the pole while the blue dashed line the branch cut}
\end{figure}
Therefore, by the Cauchy Integral Theorem:
\begin{equation*}
    0 = \oint_{\Gamma} \di z \frac{|\mathbf{p}(\mathbf{k})|}{2 \pi} \frac{e^{i\sigma rz}}{z^2 + |\mathbf{k}|^2}\hat{f}^{\nu}(|\mathbf{p}(\mathbf{k})|,\mathbf{p}(\mathbf{k})) = I - I_{pole} + I_{branch} + I_{s},
\end{equation*}
where $I$ denotes the integral we started with, $I_{pole}, I_{branch}$ the contributions of the pole and of the branch cut, while $I_s$ is the integral on the outer semicircle. In particular, by Jordan's Lemma and the fast decay assumption \eqref{eq: fast decay}, in the limit in which the radius of the semicircle is sent to infinity we have $I_s = 0$.\\
To compute the contribution of the pole, we use the Residue Theorem:
\begin{equation*}
    I_{pole} = \frac{e^{-r |\mathbf{k}|}}{2 |\mathbf{k}|} \bigg( \sqrt{|i |\mathbf{k}| \, \cos \alpha - k_1 \sin \alpha|^2 + k_2^2 + k_3^2} \bigg)\hat{f}^{\nu}(|\mathbf{p}(\mathbf{k})|, \mathbf{p}(\mathbf{k}))\bigg|_{k_0 = i |\mathbf{k}|}.
\end{equation*}
While, the branch cut contribution is:
\begin{align*}
    I_{branch} = \int_{i\frac{\sqrt{k_2^2 + k_3^2}}{\cos \alpha}}^{i \infty} \di \tau \frac{e^{i r z(\tau)}}{z(\tau)^2 + |\mathbf{k}|^2} \bigg( &\sqrt{|z(\tau) \, \cos \alpha - k_1 \sin \alpha|^2 + k_2^2 + k_3^2} \hat{f}^{\nu}(|\mathbf{p}|, \mathbf{p})\bigg|_{k_0 = z(\tau) + \epsilon} -\\
    &\sqrt{|z(\tau) \, \cos \alpha - k_1 \sin \alpha|^2 + k_2^2 + k_3^2} \hat{f}^{\nu}(|\mathbf{p}|, \mathbf{p})\bigg|_{k_0 = z(\tau) -\epsilon}\bigg),
\end{align*}
parameterized in $\tau$:
\begin{equation*}
    z(\tau) = k_1 \tan \alpha + \tau.
\end{equation*}
Performing the change of variable $\tau' = -i \tau  \cos \alpha$, we get:
\begin{equation*}
    \sqrt{|z(\tau')\cos \alpha - k_1 \sin \alpha|^2 + k_2^2 + k_3^2} = \sqrt{(\tau')^2 + k_2^2 + k_3^2} \,, \qquad \mathbf{p}(\mathbf{k}) = (i\tau', k_2, k_3).
\end{equation*}
This implies that $\sqrt{|z(\tau) \, \cos \alpha - k_1 \sin \alpha|^2 + k_2^2 + k_3^2} \hat{f}^0(|\mathbf{p}|, \mathbf{p})$, in the above integral, does not depend on $k_1$. Adding the integration over $k_1$, anyway present in \eqref{eq: integraletutto}, we get:
\begin{align*}
    \frac{i}{\cos \alpha} \int \di k_1 \int_{\sqrt{k_2^2 + k_3^2}}^{\infty} \di \tau' \frac{e^{i r k_1 \tan \alpha} e^{-r \tau'/ \cos \alpha}}{z(\tau')^2 + |\mathbf{k}|^2} \bigg( &\sqrt{|z(\tau') \, \cos \alpha - k_1 \sin \alpha|^2 + k_2^2 + k_3^2} \hat{f}^{\nu}(|\mathbf{p}|, \mathbf{p})\bigg|_{k_0 = z(\tau') + \epsilon} -\\
    & \sqrt{|z(\tau') \, \cos \alpha - k_1 \sin \alpha|^2 + k_2^2 + k_3^2} \hat{f}^{\nu}(|\mathbf{p}|, \mathbf{p})\bigg|_{k_0 = z(\tau') -\epsilon}\bigg).
\end{align*}
However, extending the integration in $k_1$ in the upper half of the complex plane and evaluating the integral in the upper semicircle $C_+$ we get contributions just from the poles located at:
\begin{equation*}
    w_{pole} = -i \tau' \sin \alpha \pm \cos \alpha \sqrt{(\tau')^2 - k_2^2 - k_3^2}.
\end{equation*}
However, as $0 < \alpha < \pi/2$, both have negative imaginary part and Residue Theorem together with the Jordan's Lemma ensure that: $\int dk_1 I_{branch} = 0$.\\
Summarizing, and using the fast decay assumption \eqref{eq: fast decay} for any $N \in \mathbb{N}$:
\begin{align*}
    \bigg| \lim_{\chi \to 0} \int \di^4 p \frac{|\mathbf{p}|}{2 \pi} \frac{e^{i(\sigma up_0 + \mathbf{p} \mathbf{x} s)}}{(p_0)^2 + |\mathbf{p}|^2}\hat{f}^{\nu}(|\mathbf{p}|,\mathbf{p}) \bigg| &= \bigg| \int \di^3\mathbf{k} \frac{e^{-r |\mathbf{k}|}}{2 |\mathbf{k}|} \bigg( \sqrt{|i |\mathbf{k}| \, \cos \alpha - k_1 \sin \alpha|^2 + k_2^2 + k_3^2} \bigg)\hat{f}^{\nu}(|\mathbf{p}|, \mathbf{p})\bigg|_{k_0 = i |\mathbf{k}|} \bigg|\\
    &\leq C_N \int \di^3\mathbf{k} \frac{e^{-r |\mathbf{k}|}}{2 |\mathbf{k}| (1 + |\mathbf{p}|)^N} \sqrt{|i |\mathbf{k}| \, \cos \alpha - k_1 \sin \alpha|^2 + k_2^2 + k_3^2} \, \bigg|_{k_0 = i |\mathbf{k}|}.
\end{align*}
In spherical coordinates:
\begin{equation*}
    |\mathbf{p}| = |\mathbf{k}| \sqrt{\cos^2 \alpha + \sin^2 \theta \, \cos^2 \varphi \, \sin^2 \alpha + \sin^2 \theta \, \sin^2 \varphi + \cos^2 \theta} \coloneqq |\mathbf{k}| A(\alpha, \theta, \varphi).
\end{equation*}
So, the above becomes:
\begin{equation*}
    \leq \frac{C_N}{2}\int \di^3 \mathbf{k} \frac{e^{-r |\mathbf{k}|}}{(1 + |\mathbf{k}|A(\alpha,\theta, \varphi))^N} A(\alpha, \theta, \varphi).
\end{equation*}
Using that $1 \leq A(\alpha, \theta, \varphi) \leq \sqrt{2}$, including all constants in a single $C \in \mathbb{R}$ and computing the angular integrals:
\begin{equation*}
    \leq C e^{r} r^{N-3} \int_{r\chi + r}^{\infty} \di x \, (x^{2-N} + r^2 x^{-N} - 2rx^{1-N}) e^{-x}.
\end{equation*}
Let us now take $\chi \to 0$, moreover from the definition of the upper incomplete Euler gamma function:
\begin{equation*}
    \Gamma(k,r) = \int_{r}^{\infty} x^{k-1} e^{-x} \di x
\end{equation*}
and the identity:
\begin{equation*}
    \Gamma(k+1,r) = k! e^{-r} \sum_{m=0}^{k}\frac{r^{m}}{m!} \, , \qquad k \in \mathbb{N}^*
\end{equation*}
restricting to $N = 0$ we get:
\begin{equation*}
    \bigg| \lim_{\chi \to 0} \int \di^4 p \frac{|\mathbf{p}|}{2 \pi} \frac{e^{i(sup_0 + \mathbf{p} \cdot \mathbf{x})}}{(p_0)^2 + |\mathbf{p}|^2}\hat{f}^{\nu}(|\mathbf{p}|,p) \bigg| \leq \frac{C}{r^3} \qquad r > 1.
\end{equation*}
Finally, combining the $r > 1$ and $r < 1$ results, we know that it exists a sufficiently big $c'_1 \in \mathbb{R}^+$ depending on $f$ such that:
\begin{equation*}
    \bigg| \lim_{\chi \to 0} \int \di^4 p \frac{|\mathbf{p}|}{2 \pi} \frac{e^{i(sup_0 + \mathbf{p} \cdot \mathbf{x})}}{(p_0)^2 + |\mathbf{p}|^2}\hat{f}^{\nu}(|\mathbf{p}|,p) \bigg| \leq \frac{c'_1}{(1+r)^{3}}.
\end{equation*}
This ends the proof.
\end{proof}
\printbibliography
\end{document}